\documentclass[authoryear,a4paper,11pt]{elsarticle_WP}

\usepackage{graphicx}
\usepackage{enumerate}
\usepackage{url}
\usepackage[tbtags]{amsmath}
\usepackage{amsfonts,amssymb,amsthm}
\usepackage{hyperref}
\usepackage{amssymb,amsmath}
\usepackage{bm}
\usepackage{bbm}
\usepackage[margin=1in]{geometry}
\usepackage{color}
\usepackage{natbib}
\usepackage{setspace}
\usepackage{booktabs}
\usepackage{array}
\usepackage[normalem]{ulem}
\usepackage{tikz}

\let\footnote=\endnote

\newcolumntype{x}[1]{>{\centering\arraybackslash\hspace{0pt}}p{#1}}

\theoremstyle{definition}
\newtheorem{definition}{Definition}

\journal{Computational Statistics and Data Analysis}

\newtheorem{theorem}{Theorem}[section]
\newtheorem{lemma}{Lemma}[section]
\newtheorem{corollary}{Corollary}[section]

\newcommand{\bD}{\mathbf{D}}
\newcommand{\bDelta}{\bs{\Delta}}
\newcommand{\bE}{\mathbf{E}}
\newcommand{\bX}{\mathbf{X}}
\newcommand{\bY}{\mathbf{Y}}

\newcommand{\zerob} {{\bf 0}}
\newcommand{\oneb} {{\bf 1}}
\newcommand{\expect} {\mathrm{E}}

\newcommand{\thetab} {{\boldsymbol{\theta}}}

\newcommand{\nub} {{\boldsymbol{\nub}}}

\newcommand{\xib} {{\boldsymbol{\xi}}}

\newcommand{\intd} {\textrm{d}}
\newcommand{\phib} {\boldsymbol{\phi}}

\newcommand{\etab} {\boldsymbol{\eta}}

\newcommand{\Amat} {\textbf{A}}
\newcommand{\Bmat} {\textbf{B}}
\newcommand{\Dmat} {\textbf{D}}

\newcommand{\Gmat} {\textbf{G}}
\newcommand{\Lmat} {\textbf{L}}
\newcommand{\Qmat} {\textbf{Q}}
\newcommand{\Rmat} {\textbf{R}}
\newcommand{\Smat} {\textbf{S}}

\newcommand{\Kmat} {\textbf{K}}
\newcommand{\Wmat} {\textbf{W}}

\newcommand{\Zmat} {\textbf{Z}}
\newcommand{\Xmat} {\textbf{X}}
\newcommand{\Xvec} {\mathbf{X}}

\newcommand{\Imat} {\textbf{I}}

\newcommand{\Pmat} {\textbf{P}}

\newcommand{\Vmat} {\textbf{V}}

\newcommand{\dvec} {\textbf{d}}

\newcommand{\xvec} {\textbf{x}}

\newcommand{\wvec} {\textbf{w}}
\newcommand{\vvec} {\textbf{v}}
\newcommand{\svec} {\textbf{s}}

\newcommand{\betab} {\boldsymbol {\beta}}

\renewcommand{\zerob}{\mathbf{0}}

\newcommand{\A}{\mathbf{A}}
\newcommand{\B}{\mathbf{B}}

\newcommand{\I}{\mathbf{I}}
\renewcommand{\P}{\mathbf{P}}
\newcommand{\R}{\mathbf{R}}
\newcommand{\Q}{\mathbf{Q}}
\renewcommand{\S}{\mathbf{S}}

\newcommand{\x}{\mathbf{x}}
\newcommand{\Yvec}{\mathbf{Y}}
\newcommand{\Wvec}{\mathbf{W}}

\newcommand{\Zvec}{\mathbf{Z}}
\newcommand{\epsilonb}{\boldsymbol{\varepsilon}}

\newcommand{\bI}{\mathbf{I}}
\newcommand{\bR}{\mathbf{R}}

\newcommand{\bB}{\mathbf{B}}
\newcommand{\bP}{\mathbf{P}}
\newcommand{\bQ}{\mathbf{Q}}

\newcommand{\var}{\mathrm{var}}

\newcommand{\tr}{\mathrm{tr}}

\newcommand{\Gau}{\mathrm{Gau}}
\newcommand{\chol}{\mathrm{chol}}

\newcommand{\bzero}{\boldsymbol{0}}

\newcommand{\plus}[1]{\mathop{+}\left\{ #1 \right\}}
\newcommand{\ctimes}[1]{\mathop{\times}\left\{ #1 \right\}}
\DeclareMathOperator{\diag}{diag}
\DeclareMathOperator{\ones}{ones}

\newcommand{\trans}{^{\scriptscriptstyle T}}

\newcommand{\ldef}{:=}
\newcommand{\bs}[1]{\boldsymbol{#1}}

\newcommand{\cE}{\mathcal{E}}
\newcommand{\cG}{\mathcal{G}}
\newcommand{\cL}{\mathcal{L}}
\newcommand{\cV}{\mathcal{V}}

\DeclareMathOperator*{\MFI}{MFI}

\let\originalleft\left
\let\originalright\right
\renewcommand{\left}{\mathopen{}\mathclose\bgroup\originalleft}
\renewcommand{\right}{\aftergroup\egroup\originalright}

\begin{document}
	
	\begin{frontmatter}
		
		\title{A sparse linear algebra algorithm for fast computation of prediction variances with Gaussian Markov random fields\tnoteref{t1}}
                \tnotetext[t1]{Reproducible code available as Supplementary Material.}		
		% \author{Andrew Zammit Mangion\footnote{National Institute for Applied Statistics Research Australia~(NIASRA), University of Wollongong, Australia} \and Noel Cressie\footnotemark[1] \and Anita L. Ganesan\footnote{School of Geographical Sciences, University of Bristol, UK} \and Matthew Rigby\footnote{Department of Chemistry, University of Bristol, UK}}
		% \date{}
		
		\author[Wollongong]{Andrew Zammit-Mangion\corref{correspondingauthor}}
		\cortext[correspondingauthor]{Corresponding author}
		\ead{azm@uow.edu.au}
		\author[Bristol]{Jonathan Rougier}
		\ead{j.c.rougier@bristol.ac.uk}
		\address[Wollongong]{National Institute for Applied Statistics Research Australia~(NIASRA), School of Mathematics and Applied Statistics (SMAS), University of Wollongong, Northfields Avenue, Wollongong, NSW 2522, Australia}
		\address[Bristol]{School of Mathematics, University of Bristol, Tyndall Avenue, Bristol, BS8 1TH, UK}
		
		\begin{abstract}
			
			Gaussian Markov random fields are used in a large number of disciplines in machine vision and spatial statistics. The models take advantage of sparsity in matrices introduced through the Markov assumptions, and all operations in inference and prediction use sparse linear algebra operations that scale well with dimensionality. Yet, for very high-dimensional models, exact computation of predictive variances of linear combinations of variables is generally computationally prohibitive, and approximate methods (generally interpolation or conditional simulation) are typically used instead. A set of conditions are established under which the variances of linear combinations of random variables can be computed exactly using the Takahashi recursions. The ensuing computational simplification has wide applicability and may be used to enhance several software packages where model fitting is seated in a maximum-likelihood framework. The resulting algorithm is ideal for use in a variety of spatial statistical applications, including \emph{LatticeKrig} modelling, statistical downscaling, and fixed rank kriging. It can compute hundreds of thousands exact predictive variances of linear combinations on a standard desktop with ease, even when large spatial GMRF models are used.
			
			%Finite-dimensional spatial models, expressed as linear combinations of spatial basis functions, are ubiquitous in studies of spatial and spatio-temporal phenomena. 
			
		\end{abstract}
		
		\begin{keyword}
			conditional dependence; GMRF; lattice spatial model; sparse inverse subset; Takahashi equations
		\end{keyword}

	\end{frontmatter}

\section{Introduction} \label{sec:Intro}

Gaussian Markov random fields (GMRFs) play a pivotal role in various applications such as image analysis \citep{Mardia_1988}, disease mapping \citep{Lawson_2011}, and atmospheric pollution modelling \citep{Cameletti_2013}. They are frequently seen as reasonable approximations to continuously-indexed Gaussian processes \citep{Rue_2002}, and are often preferred due to their favourable computational properties. Recent work on their ability to approximate Gaussian processes typically used in geostatistical models \citep[e.g.,][]{Lindgren_2011,Nychka_2015} has led to their widespread use in the space-time analysis of data at scales that were inconceivable two decades ago \citep[e.g.,][]{Zammit_2015}.

Let $\etab$ have a non-degenerate multivariate Gaussian distribution with precision matrix $\Qmat$, and encode the pairwise conditional dependence properties of $\etab$ in the form of a graph $\cG_Q = \{\cV, \cE_Q\}$, where $\cV$ indexes the elements of $\etab$, and $(i,j) \not\in \cE_Q$ exactly when $\eta_i \perp \eta_j \mid \{\eta_k: k \ne i,j \}$, for $i \neq j$.  As is well-known \citep[see, e.g.,][Theorem~2.2]{Rue_2005}, this graph defines the zeros in $\Qmat$, for which $Q_{ij} = 0$ if and only if $i \not\sim j$ in $\cG_Q$, for $i \neq j$.  As the distribution of $\etab$ is non-degenerate, the positivity condition of the Hammersley-Clifford Theorem holds \citep{Besag_1974}, and $\cG_Q$ also encodes the local Markov property and the global Markov property of $\etab$.
%Consider a set of vertices $\cV$ and a set of edges $\cE_Q = \{\{i,j\}: i \sim j; i,j \in \cV\}$, where here $\sim$ is used to denote adjacency. Then, a GMRF $\etab$ on the graph $\cG_Q \ldef (\cV,\cE_Q)$ is multivariate Gaussian with precision matrix $\Qmat,$ where $Q_{ij} = 0$ if $i \not\sim j,~i\ne j$. The matrix $\Qmat$ is required to be nonnegative-definite, but we will assume throughout that it is positive-definite. 

In this article we consider the case when the GMRF is used to encode prior belief on the quantity $\etab$ through $\Qmat$, that may itself be a function of a small number of parameters that need to be estimated. We further assume that $\etab$ is not directly observed; instead, a linear combination of $\etab$, $\Bmat\etab$, is observed in the presence of noise. Denote the data vector as $\Zvec$. The two-level hierarchical model we consider is
\begin{align*}
\Zvec &= \Bmat\etab + \epsilonb,\\
\etab &\sim \Gau(\Xvec\betab,\Qmat^{-1}),
\end{align*}
where $\Xvec$ are covariates, $\betab$ are regression coefficients, and $\epsilonb$ is Gaussian, uncorrelated, measurement error with diagonal precision matrix $\Rmat$. 

It is an immediate result that if $\Xmat,\Bmat,\Qmat,\Rmat$ and $\betab$ are known, then the precision matrix of $\etab \mid \Zvec$ is $\Pmat \ldef \Bmat\trans\Rmat\Bmat + \Qmat$. It is also well known that $\var(\eta_i \mid \Zvec)$ can be easily found from the sparse Cholesky factor of $\Pmat$ using the Takahashi equations, without computing $\Smat \ldef \Pmat^{-1}$ directly \citep{Takahashi_1973,Erisman_1975,Rue_2007}. Frequently, however, we wish to compute prediction variances of \emph{linear combinations} of $\etab$, for example over sub-groups of variables, or over regions in a spatial domain in what is sometimes referred to as the change of support problem \citep{Wikle_2005}. This computation is always needed in the ubiquitous case when the spatial field is modelled as a sum of basis functions, and where a GMRF prior is placed on the basis-function coefficients. 

This article investigates the use of sparse linear algebra methods for the computation of the marginal variances of $\Amat\etab \mid \Zvec$, that is, $\dvec \ldef \diag(\Amat \Smat \Amat\trans)$, when $\Amat$ is nonnegative and when $\Pmat$ is such that its Cholesky factor can be computed. Specifically, it establishes the conditions on $\Amat$  under which $\dvec = \diag(\Amat \tilde\Smat \Amat\trans)$, where $\tilde{\Smat}$ is a sparse subset of $\Smat$ containing a sparsity pattern that is in general identical to that of the Cholesky factor of $\Pmat$, also a by-product of the Takahashi equations. We find that in several situations of practical importance, this computation simplification facilitates the evaluation of conditional variances over linear combinations where direct  computation is only possible in a massively parallel computing environment, and where conditional simulation, while feasible, is inaccurate when the number of simulations is limited to a reasonable value.

Sparse inverse subsets are frequently used to facilitate computation in estimation frameworks \citep[e.g.,][]{Gilmour_1995,Kiiveri_2012,Cseke_2016}. They are particularly useful for computing trace operations appearing in estimating equations of the form $\dvec\trans\oneb =  \tr(\Smat\Amat\trans\Amat)$. \citet{Bolin_2009} noted that if $\Amat = \Bmat$ then $\tilde\Smat$ necessarily contains the required elements to compute the trace, and thus replaced $\Smat$ with $\tilde{\Smat}$ when computing this trace operation in the M-step of an expectation-maximisation algorithm. \citet{Vanhatalo_2010} solved the related problem of computing $\tr(\Smat\Dmat)$ where $\Dmat$ has the same sparsity pattern as $\Smat^{-1}$, by replacing $\Smat$ with $\tilde{\Smat}$. In a similar vein, \citet{Grigorievskiy_2016} computed the block-diagonal inverse subset of $\Smat$ to find the trace when both $\Smat^{-1}$ and $\Dmat$ are block tridiagonal. In this article we instead focus on the computation of all of $\dvec$, which in a spatial context are the prediction error variances at different levels of spatial aggregation (as determined by $\Amat$).

% In our case, we are more interested in computing all of $\dvec$ which yield the variances at different levels of spatial aggregation.

%noted that if $\Amat = \Bmat$ then $\tilde\Smat$ These conditions generalise the approach of \citet{Bolin_2009}, who uses the result to compute efficiently trace terms within an expectation maximisation algorithm for the case when $\Amat = \Bmat$. 

%We show that under relatively mild assumptions, $\dvec =  \diag(\Amat \tilde\Smat \Amat\trans)$, where $\tilde{\Smat}$ is a sparse subset of $\Smat$ and also a by-product of the Takahashi equations. 

Our main result is presented in Section \ref{sec:sparseinv} while a complexity analysis is given in Section \ref{sec:complexity}. In Section \ref{sec:results} we then develop the framework required for applying this result in a spatial-analysis setting, and demonstrate its use in several case studies. These studies consider conditional-autogressive models, \texttt{LatticeKrig} models, statistical downscaling, and spatial-random effects models. Section \ref{sec:conc} concludes with a brief mention of other approaches currently being investigated for when the sparse Cholesky factor is too large to compute.

\section{Main result} \label{sec:sparseinv}

Let $\A$ and $\B$ be  nonnegative matrices, and let $\Q$ and $\R$ be positive definite symmetric matrices, where the dimensions of all matrices are implicit in what follows. Define $\P \ldef \B\trans \R \B + \Q$; hence $\P$ is positive definite even if $\B$ is not full rank. Further, define $\S \ldef  \P^{-1}$.  Our objective is to compute the vector $\dvec \ldef\diag (\A \S \A\trans)$.  To summarise,
\begin{subequations}
	\begin{align}
	\P & := \B\trans \R \B + \Q, \\
	\S & := \P^{-1}, \\
	\dvec & := \diag(\A \S \A\trans) .
	\end{align}
\end{subequations}
This section presents a theorem relating $\dvec$ to the sparsity structure of $\P$ and $\A$.  %It is purely mathematical, but the following sections illustrate applications of the theorem in the context of high-dimensional inference with spatial statistical models.

The following simple Lemma establishes a necessary and sufficient condition for $\dvec$ to be invariant to any specified element of $\S$, in terms of the elements of $\A$.

\begin{lemma}\label{lem:AA}
	The vector $\dvec$ is invariant to $S_{jk}$ if and only if $[\A\trans\A]_{jk} = 0$.
\end{lemma}

\begin{proof}
	The $i$th element of $\dvec$ is
	\begin{displaymath}
	d_i = \sum_{j} \sum_{k} A_{ij} S_{jk} A_{ik} = \sum_{j} \sum_{k} (A_{ij} A_{ik}) S_{jk}.
	\end{displaymath}
	Hence $d_i$ is invariant to $S_{jk}$ if and only if $A_{ij} A_{ik} =  0$.  Therefore the entire vector $\dvec$ is invariant to $S_{jk}$ if and
	only if $A_{ij} A_{ik} = 0$ for all $i$, or, because $\A$ is nonnegative, $[\A\trans \A]_{jk} = 0$.
\end{proof}

Next we introduce a matrix-valued function whose purpose is to analyse computational sparsity.

\begin{definition}[`Ones' function]\label{def:ones}
	Let $\bD, \bD_1, \bD_2, \dots, \bD_k$ be a set of matrices of equal size, and let
	$\plus{\bD_1, \dots, \bD_k}$ denote the computed sum $\bD_1 + \cdots +
	\bD_k$ and $\ctimes{\bD_1,\cdots,\bD_k}$ denote the computed product $\bD_1\cdots\bD_k$.  Then
	\begin{subequations}
		\begin{align}
		\ones(\bD)_{ij}
		& \ldef \begin{cases}
		0 & D_{ij} = 0 \\
		1 & D_{ij} \neq 0,
		\end{cases} \label{eq:ones1} \\
		\intertext{and}
		\ones(\plus{\bD_1, \dots, \bD_k})
		& \ldef \ones( \ones(\bD_1) + \dots + \ones(\bD_k)), \label{eq:ones2} \\
                \ones(\ctimes{\bD_1, \dots, \bD_k})
		& \ldef \ones( \ones(\bD_1) \cdots \ones(\bD_k)) . \label{eq:ones3}
		\end{align} 
	\end{subequations}
\noindent We refer to $\ones(\bD)$ as the \emph{sparsity pattern} of $\bD$.
\end{definition}

%\red{The purpose of the `ones' function is to distinguish algebraic zeros from other values: if $\ones(A)_{ij} = 0$, then $A_{ij}$ is an algebraic zero.  An algebraic zero is certain to be zero, regardless of the values assigned to other elements in the calculation.  In a given numerical calculation, known zeros in the initial objects may or may not propagate all the way through to zeros in the elements of the final result.  Where they do, however, they offer the opportunity to reduce the size of the calculation by recognizing the locations of the algebraic zeros a priori.  We will use the ones function to show how zeros propagate through our calculation of $\dvec$.}

Notice the need to clarify the difference between a \emph{computed} operation and an \emph{algebraic} operation.  The `ones' function processes a computed operation and returns a $1$ for each element of the result that needs to be evaluated and stored.  This definition ensures that zeros in an algebraic operation, which happen to arise though a lucky combination of nonzero elements cancelling out, are represented as $1$'s in the sparsity pattern of a computed operation.  Elements which are zero in the sparsity pattern of a computed operation are known as \emph{structural zeros} while elements that are zero in the sparsity pattern of an algebraic operation but not in that of a computed operation are known as \emph{algebraic zeros}. Below are some useful properties of the `ones' function.

\begin{lemma}\label{res:ones}
	Interpreting all binary relations elementwise, for any matrices $\bD,\bE,\bX,\bY$ with compatible dimensions,
	\begin{enumerate}[A.]
		
		\item If $c \neq 0$, $\ones(c \bD) = \ones(\bD)$.
		
		\item $\ones(\plus{\bD, \bzero}) = \ones(\bD)$ and $\ones(\plus{\bD, \bE}) \geq
		\ones(\bD)$.
		
		\item If $\ones(\bX) \geq \ones(\bY)$, then
		\begin{enumerate}[i.]
			
			\item $\ones(\plus{\bX, \bD}) \geq \ones(\plus{\bY, \bD})$
			
			\item $\ones(\ctimes{\bD,\bX,\bE}) \geq \ones(\ctimes{\bD,\bY,\bE})$.
			
		\end{enumerate}
                \item If $\bX$ and $\bY$ are nonnegative, then $\ones(\bX\trans\bY) = \ones(\ctimes{\bX\trans,\bY}).$
		%	\item  If $\D$ is nonnegative, then $\ones(\D\trans\D) = \sum_{i} \ones((\D_{i:})\trans\D_{i:})$,where $\D_{i:}$ denotes the $i$th row of $\D$. 
	\end{enumerate} 
\end{lemma}

\begin{proof}
	(A) and (B) are straightforward.  For (Ci),
	\begin{align*}
	\ones(\plus{\bX, \bD})
	& = \ones(\ones(\bX) + \ones(\bD)) \\
	& = \ones(\ones(\bY) + \bDelta + \ones(\bD)) \\
	& \geq \ones(\ones(\bY) + \ones(\bD)) \\
	& = \ones(\plus{\bY, \bD}) ,
	\end{align*}
	where $\bDelta$ is some matrix with $\Delta_{ij} \in \{0, 1\}$.
	For (Cii),
	\begin{align*}
	\ones(\ctimes{\bD,\bX,\bE})_{ij}
        & = \ones(\ones(\bD)\ones(\bX)\ones(\bE))_{ij} \\
	& = \ones \left( \sum_k \sum_l \ones(\bD)_{ik} \ones(\bX)_{kl} \ones(\bE)_{lj} \right) \\
	& = \ones \left( \sum_{k : D_{ik} \neq 0} \, \sum_{l : E_{lj} \neq 0}  \ones(\bD)_{ik} \ones(\bX)_{kl} \ones(\bE)_{lj} \right) \\
	& = \ones \left( \sum_{k : D_{ik} \neq 0} \, \sum_{l : E_{lj} \neq 0}  \ones(\bX_{kl}) \right) \\
	& \geq \ones \left( \sum_{k : D_{ik} \neq 0} \, \sum_{l : E_{lj} \neq 0} \ones(\bY_{kl}) \right) \\
	& = \ones(\ctimes{\bD,\bY,\bE})_{ij},
	\end{align*}
	after reversing the steps. For (D), it follows from the nonnegativity of $\bX$ and $\bY$ that $(\bX\trans \bY)_{ij} = 0 	\Longleftrightarrow  X_{ki}Y_{kj} = 0$ for all k.  This holds if and only if $\sum_k \ones(\bX)_{ki} \ones(\bY)_{kj} = 0$, or, equivalently, $\ones(\ctimes{\bX\trans, \bY})_{ij} = 0$.

        % For (D) we have that
	% $$
	% [\D_{i:}\trans\D_{i:}]_{jk} \ne 0 \iff D_{ij}D_{ik} \ne 0.
	% $$
	% Therefore, since $\D$ is nonnegative,
	% $$
	% \sum_{i}[(\D_{i:})\trans\D_{i:}]_{jk} \ne 0 \iff \sum_{i} D_{ij}D_{ik} \ne 0 \iff [\D\trans\D]_{jk} \ne 0.
	% $$
\end{proof}

Now, let  
$$
\Pmat^c := 	\plus{\ctimes{\bB\trans, \bR, \bB}, \bQ},
$$
be the \emph{computed} version of $\Pmat$. Recall from Definition \ref{def:ones} that although $\Pmat$ and $\Pmat^c$ are identical algebraically, $\ones(\Pmat^c) \ne \ones(\Pmat)$ in general. Also, $\ones(\Pmat^c) \ge \Imat$ since $\Qmat$ is positive-definite. Let $\cG_{P^c} := (\cV, \cE_{P^c})$ be the graph corresponding to the adjacency matrix $\ones(\Pmat^c) - \Imat$ and define
$$
F(i,j) := \{i+1,\dots,j-1,j+1,\dots\},\quad i < j \textrm{~and~} i,j\in\cV,
$$
to be the indices corresponding to the `future' of $i$ not including $j$. Following \citet[][Section 2.4.1]{Rue_2005}, 
% construct a new graph $\cG_{L^s + L_s\trans} := (\cV, \cE_{L^s + L_s\trans})$ such that 
% $$
%  \cE_{L^s + L_s\trans} := \{\{i,j\}: j > i~\textrm{and}~F(i,j)~\textrm{does not separate}~i~\textrm{and}~j~\textrm{in}~\cG_{P^c}   \}.
% $$
% Let $\Lmat^s$ be the lower triangle of the adjacency matrix corresponding to  $\cG_{L^s + L_s\trans}$ plus the identity matrix.
construct a lower-triangular matrix $\Lmat^s$ such that
$$
L^s_{ji} = \begin{cases}
1 \quad \textrm{if~} (j = i) \textrm{~or~} (j > i ~\textrm{ and }~ \textrm{$F(i,j)$ separates $i$ and $j$ in $\cG_{P^c}$}),\\
0 \quad \textrm{otherwise.}
\end{cases}
$$
The matrix $\Lmat^s$ then contains ones at entries that are non-zero in general (i.e., structural non-zeros) in the Cholesky factor of \emph{any} positive definite matrix with sparsity pattern identical to $\ones(\Pmat^c)$.  We refer to the matrix $\Lmat^s$ as the \emph{symbolic} Cholesky factor of $\ones(\Pmat^c)$. If $\Lmat$ is the lower-triangular Cholesky factor of $\bP$, then $\Lmat^s \ge \ones(\Lmat)$.

%Let $\Lmat$ be the Cholesky factor of $\bP$, that is, the unique lower-triangular matrix for which $\Lmat \Lmat\trans = \bP$. 
%The corresponding \emph{symbolic} factor of $\bP$, $\Lmat^s$, is the lower-triangular matrix containing ones at entries where $\Lmat$ is not known to be zero prior to its computation, and zeroes elsewhere. 

We next present a key result, which provides conditions on $\Q$, $\R$, $\A$ and $\B$ under which invariance of $\dvec$ to $S_{jk}$ can be deduced from another matrix.

\begin{lemma}\label{lem:LLP}
	Let $\Lmat^s$ be the symbolic Cholesky factor of $\ones(\bP^c)$. If
	\begin{equation}\label{eq:ineq1}
	\ones(\plus{\ctimes{\bB\trans, \bR, \bB}, \bQ}) \geq \ones(\Amat\trans \Amat),
	\end{equation}
	then $\dvec$ is invariant to $S_{jk}$ when $L_{jk}^s = 0$ and $j \ge k$.
\end{lemma}

\begin{proof}
	It is a straightforward result that non-zero elements in the lower-triangular part of $\ones(\bP^c)$ are also non-zero in $\Lmat^s$, that is, $L^s_{jk} \geq \ones(\bP^c)_{jk}$ when $j \geq k$ \citep[e.g.,][Thereom~4.2]{Davis_2006}.  Therefore, from \eqref{eq:ineq1},
	\begin{equation}\label{eq:Lsineq}
	L^s_{jk} \geq \ones(\bP^c)_{jk} \geq \ones(\Amat\trans \Amat)_{jk}, \qquad \textrm{when~} j \geq k.
	\end{equation}
	So if $L^s_{jk} = 0$ then $(\Amat\trans \Amat)_{jk} = 0,$ when $j \ge k$, in which case, by
	Lemma~\ref{lem:AA}, $\dvec$ is invariant to $S_{jk}$.  
	
\end{proof}

This leads us to the following corollary, on which Lemmas \ref{lem:Case1} and \ref{lem:Case2} found below are based.

\begin{corollary}\label{cor:conditions}
	If either
	\begin{enumerate}[i.]
		\item $\ones(\bB\trans\bB) \geq \ones(\Amat\trans \Amat)$, or
		\item $\ones(\bQ) \geq \ones(\Amat\trans\Amat)$,
	\end{enumerate}
	then $\dvec$ is invariant to $S_{jk}$ when $L_{jk}^s = 0$ and $j \ge k$.  As a special case of (ii), if $\ones(\Amat)$ is a permutation matrix, then $\dvec$ is invariant to $S_{jk}$ when $L_{jk}^s = 0$ and $j \ge k$.
\end{corollary}

\begin{proof}\hspace{1ex}
	
	\begin{enumerate}[i.]
		
		\item Since $\bR$ is positive definite, $\ones(\bR) \geq \ones(\bI)$, which implies by Lemma~\ref{res:ones} (Cii) and (D) that $\ones(\ctimes{\bB\trans, \bR, \bB}) \geq \ones(\ctimes{\bB\trans,\bB}) = \ones(\bB\trans \bB)$ since $\bB$ is nonnegative. The condition $\ones(\bB\trans \bB) \geq \ones(\Amat\trans \Amat)$	then implies \eqref{eq:ineq1}, by Lemma~\ref{res:ones} (B).
		
		\item Follows directly from Lemma~\ref{res:ones} (B).
		
	\end{enumerate}
	If $\ones(\Amat)$ is a permutation matrix, then $\Amat$ has only a single nonzero element per row and a single nonzero element per column. Therefore $\Amat\trans \Amat$ is diagonal, so that $\ones(\bQ) \geq \ones(\Amat\trans \Amat)$, because $\bQ$ is positive definite.  
\end{proof}

We now present sufficient conditions for the two cases in Corollary \ref{cor:conditions} to hold. %hows that it is sufficient that, for this condition to be satisfied, either $\ones(\B\trans\B) \ge \ones(\A\trans\A)$ or $\ones(\Q) \ge \ones(\A\trans\A)$. We consider these two cases separately.

\paragraph{Case 1 $(\ones(\B\trans\B) \ge \ones(\A\trans\A))$} The conditions for when this is true are given by the following lemma.

\begin{lemma}\label{lem:Case1}
	If, for every $j,k$,
	\begin{equation}\label{eq:cor1}
	\exists i' : A_{i'j} A_{i'k} > 0 \implies \exists i : B_{ij} B_{ik} > 0,
	\end{equation}
	then $\ones(\bB\trans \bB) \geq \ones(\Amat\trans \Amat)$.
\end{lemma}

\begin{proof}
	Recall that $\Amat$ and $\Bmat$ are nonnegative. Equation \eqref{eq:cor1} is then the condition under which the inequality below holds:
	\begin{align*}
	\ones(\bB\trans \bB)_{jk}
	& = \ones \left( \sum_i B_{ij} B_{ik} \right) \\
	& \geq \ones \left( \sum_{i'} A_{i'j} A_{i'k} \right) \\
	& = \ones(\Amat\trans \Amat)_{jk} . \qedhere
	\end{align*} 
\end{proof}

\paragraph{Case 2 $(\ones(\Q) \ge \ones(\A\trans\A))$} The conditions for when this is true are given by the following lemma.

\begin{lemma}\label{lem:Case2}
    If, for every $j,k$,
	\begin{equation}\label{eq:cor2}
	\exists i : A_{ij} A_{ik} > 0 \implies Q_{jk} \neq 0,
	\end{equation}
	then $\ones(\bQ) \geq \ones(\Amat\trans \Amat)$.
\end{lemma}

\begin{proof}
	Recall that $\Amat$ is nonnegative. Equation \eqref{eq:cor2} is the condition under which the inequality below holds: 
	\begin{align*}
	\ones(\Amat\trans \Amat)_{jk}
	& = \ones\left( \sum_i A_{ij} A_{ik} \right) \\
	& \leq \ones( Q_{jk} ) \\
	& = \ones(\bQ)_{jk} . \qedhere
	\end{align*}
\end{proof}

Combining Lemmas \ref{lem:Case1} and \ref{lem:Case2} with Corollary \ref{cor:conditions} we obtain the following theorem.

\begin{theorem}\label{thm:LLP}
	If, for every $j,k$,
	\begin{equation}\label{eq:cor3}
	\exists i' : A_{i'j} A_{i'k} > 0 \implies \exists i : B_{ij} B_{ik} > 0 \textrm{~or~} Q_{jk} \ne 0,
	\end{equation}
	then $\dvec$ is invariant to $S_{jk}$ when $L_{jk}^s = 0$ and $j \ge k$.
\end{theorem}

%\begin{definition}
%The $(j,k)$th elements of $\Smat$ for which $L_{jk} \ne 0,$ are termed the \emph{sparse inverse subset} of $\Smat$, and the matrix containing only these elements is termed the \emph{sparse inverse} of $\Pmat = \Smat^{-1}$.
%\end{definition}

%In practice, there may be cases where the required condition is violated in Theorem \ref{thm:LLP}. In this case all that is required is the treatment of certain elements as nonzero in $\P$, even if they are numerically zero. These elements correspond to introducing `dummy cliques' in $\Q$, that is, cliques where two or more variables might be conditionally independent. These dummy cliques need to be introduced until $\cA_i \subseteq \cQ, i = 1,\dots,N$. There is a point where the treatment of many zeroes as nonzeroes becomes wasteful (e.g., when linear combinations over a large number of variables are required) in which cases resorting to standard approaches becomes once again feasible.

\subsection*{Remarks}

%Here we give an application of the theory presented above, to provide a better understanding of the computational challenge, and intuition concerning the sufficient conditions that support the results.

Suppose, as detailed in Section \ref{sec:Intro}, that we wish to compute the conditional variances of linear combinations of $\etab$, that is, $\dvec \ldef \diag(\var(\Amat \etab | \Zvec))$. By Theorem~\ref{thm:LLP}, we see that in some cases we only need those $S_{jk}$ for which $L_{jk}^s \ne 0$ to compute $\dvec$. However, the conditions in Theorem \ref{thm:LLP}, in all their generality, are unwieldy, and in practice it is easier to ensure that the conditions of Corollary \ref{cor:conditions} hold. These conditions, although still sufficient, are more restrictive than those in Theorem \ref{thm:LLP}. Two sufficient conditions are given:

\begin{enumerate}
	
	\item  First, \eqref{eq:cor1} provides a sufficient condition for $\ones(\bB\trans \bB) \geq \ones(\Amat\trans \Amat)$.  It states that every pair of nonzero elements that appear in a row of $\Amat$ must also be nonzero in a row of $\Bmat$.  When $\etab$ corresponds to a GMRF over a tiling of a spatial domain, this condition holds when the observation footprints encoded in the rows of $\Bmat$ are such that they tile the domain of interest into large regions, and that the prediction footprints encoded in the rows of $\Amat$ are such that they tile the domain into smaller-footprint regions that are \emph{nested} inside the large-footprint regions. Satisfying this sufficient condition might be a fortuitous outcome of the experiment design, for example as shown in the case study of statistical downscaling in Section \ref{sec:downscaling}.
	
	%Generally, the analyst does not have much control over $\Amat$ and $\Bmat$: the former is set by the observation campaign and the latter by the client.
	% So  but it is not in the analyst's control.
	
	\item Second, \eqref{eq:cor2} provides a sufficient condition for $\ones(\bQ) \geq \ones(\Amat\trans \Amat)$.  It states that every pair of vertices that appear in a row of $\Amat$ must correspond to an edge in $\cG_Q$, and must hence be a nonzero value in $\bQ$. Specifically, the nonzero vertices in each row of $\Amat$ must be a clique in $\cG_Q$. When $\etab$ corresponds to a GMRF over a tiling of a spatial domain, this sufficient condition will be satisfied if all predictions are for regions no larger than those implied by the neighbourhood structure of $\cG_Q$.  
	
\end{enumerate}

Recall from Corollary \ref{cor:conditions} that a special case of 2.~occurs when $\ones(\Amat)$ is a permutation matrix. Hence, it follows that care as to whether conditions of Theorem \ref{thm:LLP} are satisfied is only needed when we want to predict over linear combinations of $\etab$. If neither 1.~or 2.~are satisfied, the analyst can intervene by treating elements that are zero in the sparse matrices as being structurally non-zero. For the second case, this means adding cliques to $\cG_Q$ by forcing selected elements in $\ones(\bQ)$ to be non-zero.  Forcing $\ones(Q_{jk})$ to be non-zero (even though $Q_{jk} = 0$) ensures that $L^s_{jk} \neq 0,$ and hence that the inequality \eqref{eq:Lsineq} is satisfied. %Hence, if a prediction region is a little larger than a neighbourhood structure in $\cG$, then $\bQ$ can be padded with explicit zeros which are not edges in $\cG$, but which ensure that \eqref{eq:cor2} holds.
Whether this is a good strategy for computation will depend on how many elements need to be coerced, a point we explore further in the next section.

The $\ones$ function and the smybolic Cholesky factor are essential for the result of Theorem \ref{thm:LLP} to hold. This means that, when producing software, one must ensure that elements that are structurally non-zero are contained within the matrix objects, whether or not these are algebraically computed to be zero or not. No two software packages are the same, and while the \texttt{chol} function in the \texttt{Matrix} package in \texttt{R} \citep{R} retains the structural non-zeros when computing the Cholesky factor, the \texttt{chol} function in \texttt{MATLAB} does not. In \texttt{MATLAB}, the function \texttt{symbfact} can be used to compute $\Lmat^s$.

Finally, the conditions of Lemma \ref{lem:LLP} and Theorem \ref{thm:LLP} dictate when $\dvec \ldef \diag(\A \S \A\trans) = \diag(\A \tilde\S \A\trans)$, where $\tilde\S$ is a sparse inverse subset of $\P$, given by
$$
\tilde{S}_{jk} = \tilde S_{kj}  =\begin{cases}
S_{jk} & L^s_{jk} \ne 0,\quad j\ge k \\
0 & \textrm{otherwise}.    
\end{cases}
$$
Therefore, these results are only useful because there is an efficient way of computing the sparse inverse subset $\tilde{\S}$, that is, the elements $S_{jk}$ for which $L^s_{jk}\ne 0$. This computation is based on the Takahashi equations, which we discuss in more detail in the next section. 

\section{Computational Complexity} \label{sec:complexity}

%Let $\Lmat$ be the Cholesky factor of $\bP$, that is, the unique lower-triangular matrix for which $\Lmat \Lmat\trans = \bP$. 
In the following discussion we assume that $\P$ has been permuted such that the bandwidth of its lower-triangular Cholesky factor, $\Lmat$, has been minimised \citep[][Section 2.4.2]{Rue_2005}; focussing on the banded case renders the complexity analysis straightforward, but our experience so far has been that one can draw similar conclusions when other fill-in reducing permutations are used instead. Permutation does not affect our conclusions in Section \ref{sec:sparseinv}, which do not depend on any specific ordering on $\cV$. We denote the bandwidth of $\Lmat$ as $b$, let $\A$ be of size $N \times n$, and let $\Pmat$ be of size $n \times n$.

Conventionally, the vector $\dvec$ is computed directly from $\Lmat$ and $\Amat$ by first solving $\Lmat\Gmat = \A\trans$, and then computing $\dvec = (\Gmat\trans \circ \Gmat\trans)\oneb$, where $\circ$ is the Hadamard (elementwise) product. The number of multiplicative operations required to compute the forwardsolve is at most $N$ times the number of nonzero elements in $\Lmat$;  if the band of $\Lmat$ is full, then $N(b+1)(n - b/2)$ operations are required \citep[][Lemma 2.2.1]{George_1981}. In practice the number of operations will be less when the sparsity structure of $\Gmat$ is taken into account before computation. The number of multiplicative operations for the Hadamard product and column-summation is at most $Nn$. Therefore, the total number of multiplicative operations required, given $\Lmat$,  is at most
$$
N(b+1)(n - b/2) + Nn,
$$
which is $O(Nbn)$. 

Now, assume that the sparse inverse subset of $\P$, denoted as $\tilde\S$, is known. Under the conditions of Theorem \ref{thm:LLP}, the vector $\dvec$ can be found efficiently by computing $(\A \circ (\A\tilde\S))\oneb$. If $\Lmat$ has bandwidth $b$, then so does $\tilde\S$, so that the number of operations required for the inner matrix multiplication is $N(b+1)n_\#$, where $n_\#$ is the number of nonzeros in each row of $\A$ (assuming, for simplicity, that this is a constant). The Hadamard product requires at most $Nn_\#$ multiplicative operations, so that the total number of multiplicative operations required when using $\tilde\S$ to compute $\dvec$ is at most
$$
N(b+1)n_\# + Nn_\#,
$$
which is $O(Nbn_\#)$. Therefore, if $\tilde\S$ is known, using it to find the predictive variances will always be faster than the direct method, and frequently dramatically so.

It therefore now remains to analyse the cost of computing the Takahashi equations \citep{Takahashi_1973} required for computing $\tilde\S$. These equations are reproduced below for completeness \citep[see also][Equation 7]{Rue_2009},
\begin{align*}
\textrm{off-diagonal terms:} & \quad \tilde S_{ij} = \tilde S_{ji}  = -\frac{1}{L_{ii}}\sum_{k = i+1}^nL_{ki}\tilde S_{kj},~~~ j > i,~~ i = n,\dots,1, \\ 
\textrm{diagonal terms:} &~~~~~~~~ \quad\tilde S_{ii} = \frac{1}{L_{ii}^2}- \frac{1}{L_{ii}}\sum_{k = i+1}^nL_{ki}\tilde S_{ki},~~ i = n,\dots,1.
\end{align*}
Each Takahashi recursion requires at most $b+1$ multiplicative operations ($b+2$ when computing $\tilde{S}_{ii}$, which we ignore for convenience). The total number of multiplicative operations required is then the product of the total number of elements that need to be computed (which in this case equals the number of nonzero elements in $\Lmat$) by $(b+1)$,
$$
(b+1)^2(n - b/2),
$$
which is $O(b^2n)$. 

Both the direct approach and the sparse inverse subset approach require computation of $\Lmat$ which is also $O(b^2n)$ \citep[][Theorem 2.1.2]{George_1981}. We thus have the following two algorithmic costs for when $\Lmat$ is banded:
\begin{align*}
\textrm{Direct cost :} &\quad O(\max(Nbn, b^2n)), \\
\textrm{Sparse inverse subset approach cost:}& \quad O(\max(Nbn_\#, b^2n)).
\end{align*}

\noindent Therefore, even after taking the Takahashi equations into account, we see that use of $\tilde\S$ is still very attractive when $n_\# \ll n$. This corresponds to cases when predictions are required at points or over small areas. For images usually $n_\# = 1$ pixel, while in two-dimensional spatial applications $n_\#$ is usually less than 10.

For both methods it is critical that $b$ is `not too large,' but even more so when using the sparse inverse subset, in which case the additional time required to compute the Takahashi equations has a deleterious effect on performance. In practice, a large bandwidth could be the result of having (i) large prediction regions, (ii) a high-order Markov assumption when constructing $\Q$, (iii) a measurement error dependence structure so that $\R$ is not sparse, and (iv) measurements with large observation footprints, so that $\B$ is relatively dense. Deliberately forcing elements in $\ones(\Qmat)$ to be non-zero so that Corollary \ref{cor:conditions} (ii) is satisfied (see \emph{Remarks} in Section \ref{sec:sparseinv}) will thus also have a negative impact.  We observed that (i) to (iv) are important considerations even when other fill-in reducing permutations are used.

%There are two important remarks are warranted. First, although both the banded sparse Cholesky decomposition and Takahashi equations scale as $O(nb^2)$, the constant of proportionality is different. Specifically, the number of computations required for a banded Cholesky decomposition when $b \ll n$ is approximately $b^2n/2 + 3bn/2$ while the Takahashi equations need approximately $2b^2n + 2bn$ operations, roughly a factor of 4 more. 

The above analysis considers the total number of operations, but many of these can be parallelised. In particular, in the direct approach one can brute-force computation of $\dvec$ by computing the columns of $\Gmat$ in parallel.  Once $\Lmat$ has been computed, the direct approach is then $O(Nbn/\alpha)$, where $\alpha$ is the number of parallel processing units available. When $\alpha$ is high, direct computation becomes once again preferable to use of the sparse inverse subset, which remains $O(b^2n)$ because of the sequentially-computed Takahashi equations. Achieving a high $\alpha$ is possible with judicious use of graphical processing units, which we do not consider in this article.

\section{Studies on applications in spatial statistics}\label{sec:results}

In this section we apply the theory of Section \ref{sec:sparseinv} to reduced-rank modelling in spatial statistical applications, motivated in Section \ref{sec:spatial}. We then study the potential benefit of using the sparse inverse subset for computing the prediction variance in various contexts. First, we consider a second-order conditional autoregressive (CAR) model in one dimension to verify the computational benefit; second, we show its utility in improving prediction-error estimation in the package \texttt{LatticeKrig}; third, we show its use in the problem of statistical downscaling; finally, we show how it greatly simplifies the problem of prediction-error computation in a variant of fixed rank kriging that includes a conditional auto-regressive model to capture fine-scale variation. All reported timings are based on scripts written in \texttt{R} run on a platform using OpenBlas \citep{Wang_2013} and an Intel\textregistered~Core\texttrademark~i7-4712HQ 2.30 GHz processor using a single core (while noting that all methods described below can be parallelised to various degrees). The \texttt{SuiteSparse} libraries were used to implement fill-in reducing permutations and the Takahashi equations \citep{SuiteSparse}. Code and data for all analyses and results shown are available as supplementary material.

\subsection{Application to spatial statistics}\label{sec:spatial}

Spatial modelling and prediction approaches that employ finite-dimensional representations of infinite-dimensional stochastic processes generally fall into two classes: Covariance-centric \citep[e.g.,][]{Cressie_2008} and precision-centric \citep[e.g.,][]{Lindgren_2011} approaches. In this article we are concerned with the latter, where spatial dependence is characterised through the conditional dependence structure of either (i) the coefficients of basis functions used to decompose the field or (ii) the field averaged over spatial regions. We are further concerned with the ubiquitous case of when the measurement error is Gaussian and parameter inference is done within a maximum-likelihood framework \citep[e.g.,][]{Nychka_2015}. When employing a maximum-likelihood framework, all predictions following parameter fitting are also conditionally Gaussian when conditioned on the data, however their computation is not necessarily straightforward: Prediction might be needed at millions of locations when large datasets are used and fine-resolution maps are needed. 

The general setting is the following. Consider a spatial field $Y(\svec), \svec \in D \subset \mathbb{R}^d$. Then, under the assumption that the space of $Y$ is spanned by a set of known basis functions, its finite-dimensional representation is given by $$Y(\svec) = \xvec(\svec)\trans\betab + \phib(\svec)\trans\etab,\quad \svec \in D,$$ 
where $\xvec(\cdot)$ is a vector of $p$ specified regression functions, $\betab$ are regression coefficients, $\phib(\cdot)$ is a vector of $n$ basis functions that model spatial structure not captured by the regressors, and $\etab$ are the basis-function coefficients. Here we consider the case when $\etab \sim \Gau(\zerob,\Qmat^{-1})$ and $\Qmat$ is sparse, that is, when $\etab$ is a GMRF.

The process $Y(\cdot)$ is observed through some data $Z_k, ~k = 1,\dots,m$. Crucially, we assume that each $Z_k$ is a specified linear operator on $Y$, plus a measurement error, that is,
\begin{equation}
Z_k = \cL_k^O \cdot Y + \epsilon_k,\quad k = 1,\dots,m.
\end{equation}
Then, it follows that 
\begin{equation}
\Zvec = \Xmat^O\betab + \Bmat\etab + \epsilonb,
\end{equation}
where $X^O_{ki} \ldef \cL_k^O\cdot x_i$, $B_{ki} \ldef \cL_k^O \cdot \phi_i$. The vectors $\Zvec$ and $\epsilonb$ are the vectors containing the $m$ data and measurement errors, respectively.

The $N$ predictions $Y_j,~j = 1,\dots, N$, are also assumed to be specified (potentially different) linear operators on $Y$, that is,
\begin{equation}
Y_j = \cL_j^P \cdot Y,~j = 1,\dots,N.
\end{equation}
Then, it follows that 
\begin{equation}
\Yvec = \Xmat^P\betab + \Amat\etab,
\end{equation}
where $X^P_{ji} \ldef \cL_j^P \cdot x_i$, $A_{ji} \ldef \cL_j^P \cdot \phi_i$.  Hence prediction reduces to the computation of the prediction mean, $\hat{\Yvec} \ldef \Xvec^P\betab + \A  \expect(\etab \mid \Zvec)$, and the prediction variances
\begin{equation}\label{eq:YpZ}
\dvec \ldef \var(\Yvec \mid \Zmat) = \diag( \A \S \A\trans), 
\end{equation}
\noindent where
\begin{equation}\label{eq:S}
\S \ldef \var(\etab \mid \Zvec) = (\Bmat\trans\Rmat\Bmat + \Qmat)^{-1}.
\end{equation}

The computation of $\dvec$ can be challenging, and a common alternative is to simulate from $\etab \mid \Zvec$ and obtain prediction standard errors using sample variances. In Section \ref{sec:sparseinv} we derived the conditions under which we could used the sparse inverse subset $\tilde{\Smat}$ instead of $\Smat$ in \eqref{eq:YpZ} to compute $\dvec$, and in Section \ref{sec:complexity} discussed the respective computational properties. In the following sections we carry out simulation studies and show that using $\tilde{\Smat}$, where valid, is much more efficient than direct methods of computation and, sometimes, also conditional simulation.

\subsection{Simulation studies on computational efficiency}\label{sec:simstudy}

Assume $D = [0,1]$ and that the process $Y(s), s \in D,$ has zero expectation and can be decomposed into $n$ bisquare basis functions of the form
\begin{equation}\label{eq:bisquare1D}
\phi(s;s^c_i,r_n) \ldef \left\{\begin{array}{ll} \{1 - (|s - s^c_i|/r_n)^2\}^2; &| s  - s^c_i| \le r_n \\ 
0; & \textrm{otherwise}, \end{array} \right. 
\end{equation}
where $s^c_i, i = 1,\dots,n,$ are the function centroids which are equally spaced in $D$, and $r_n$ is the aperture, taken to depend on the number of basis functions $n$. Since $Y$ has zero expectation, $Y(\cdot) = \phib(\cdot)\trans\etab$, where $\etab \sim \Gau(\zerob,\Qmat^{-1})$. As a model for $\etab$ we let $\Qmat = \tau(\Imat - \rho\Wmat)$ with $\rho = 1/12$ and $\tau = 12$, where $\Wmat$ is the proximity matrix
$$
W_{ij} = \begin{cases} 
4 & |j-i| = 1 \\
1 & |j-i| = 2 \\
0 & \textrm{otherwise};
\end{cases}
$$
\noindent see \citet[Section 4.3]{Banerjee_2015} for further details. In this model, $\Qmat$ has at most $5$ nonzero elements per row (note that only the sparsity pattern for $\Qmat$ is important for this study, not the actual values).

Assume now that we have observed $Y(\cdot)$ at $m$ random locations chosen uniformly on $D$ with measurement-error variance 0.1 (hence $\R = 10\I$). Since the observations are point referenced, we take $\cL^O_k, k = 1,\dots,m,$ to be identity operators. The matrix $\Bmat$ is hence the $m \times n$ matrix containing evaluations of the basis functions at the observation locations. As prediction domain we use $N$ points equi-spaced in $[0,1]$. Since this implies that $\cL^P_j,j = 1,\dots,N$, are identity operators, the matrix $\A$ is just the $N \times n$ matrix containing evaluations of the basis functions at the prediction locations.

The bisquare basis-function centroids are equi-spaced in $[0,1]$ (see Figure~\ref{fig:exp_setup}), and we set $r_n = 1/n$ so that at any point in $D$ only one or two of the basis functions evaluate to 0. Then, $\Amat$ contains at most two nonzeroes per row, corresponding to two basis-function weights that are \emph{not} conditionally independent due to the use of the proximity matrix $\Wvec$ in constructing $\Qmat$. Condition \ref{eq:cor2} is satisfied and therefore, by Corollary \ref{cor:conditions}, (ii) we can use the sparse inverse subset to compute $\dvec$.

The aim of this experiment is to see how the computational time required to compute prediction variances varies when using different approaches for different $n$ and $N$. Specifically, we let $n = 10^2,\dots,10^5$ and $N = 10^1,\dots,10^5$. For each combination of $n$ and $N$ we evaluate the time needed to compute the predictive variances given in \eqref{eq:YpZ} and \eqref{eq:S}. Since processing times increase linearly with $m$ in low-rank methods, $m$ is not a factor we vary in the experiment, and fix it to $10^4$. 

We consider the following methods for computing the prediction variances:
\begin{enumerate}
	\item {\bf Direct (exact):} The `standard approach' obtains variances by
	\begin{enumerate}
		\item Finding the Cholesky factor $\Lmat = \chol(\Smat^{-1})$.
		\item Solving for $\Gmat$ in the forward substitution $\Lmat\Gmat = \A^{\trans}$.
		\item Evaluating row sums over a Hadamard product to find $\dvec = (\Gmat\trans\circ\Gmat\trans)\oneb$.
	\end{enumerate}
	\item {\bf Sparse inverse subset (exact):} The `sparse inverse subset' approach obtains variances by
	\begin{enumerate}
		%\item Checking whether the conditions of Theorem 2.1 are satisfied. If not, `dummy cliques' are added to $\mathcal{Q}$ (the set of cliques of $\Smat^{-1}$) to satisfy the conditions
		\item Finding the Cholesky decomposition $\Lmat = \chol(\Smat^{-1})$.
		\item Using the Takahashi equations to compute the sparse inverse subset $\tilde\Smat$.
		\item Evaluating row sums over a Hadamard product to find $\dvec = (\A \circ (\A\tilde\S))\oneb$.
	\end{enumerate}
	\item {\bf Conditional simulation (approximate):} Conditional simulation can be used to obtain approximate prediction variances by 
	\begin{enumerate}
		\item Finding the Cholesky decomposition $\Lmat = \chol(\Smat^{-1})$.
		\item Solving by backwards substitution $\Lmat^{\trans} \vvec_i = \wvec_i, i = 1,\dots,M$, where $\wvec_i \sim \Gau(\zerob,\Imat)$ and we fix $M = 50$.
		\item Finding the empirical prediction variance from the simulations $\{\vvec_i: i = 1,\dots,M\}$.
	\end{enumerate}
\end{enumerate}
\noindent In all the above we use  permutations of $\Smat^{-1}$ in order to reduce the number of fill-ins when computing $\Lmat$.
\begin{figure}[!t]
	\begin{center}
		\includegraphics[width=0.9\textwidth]{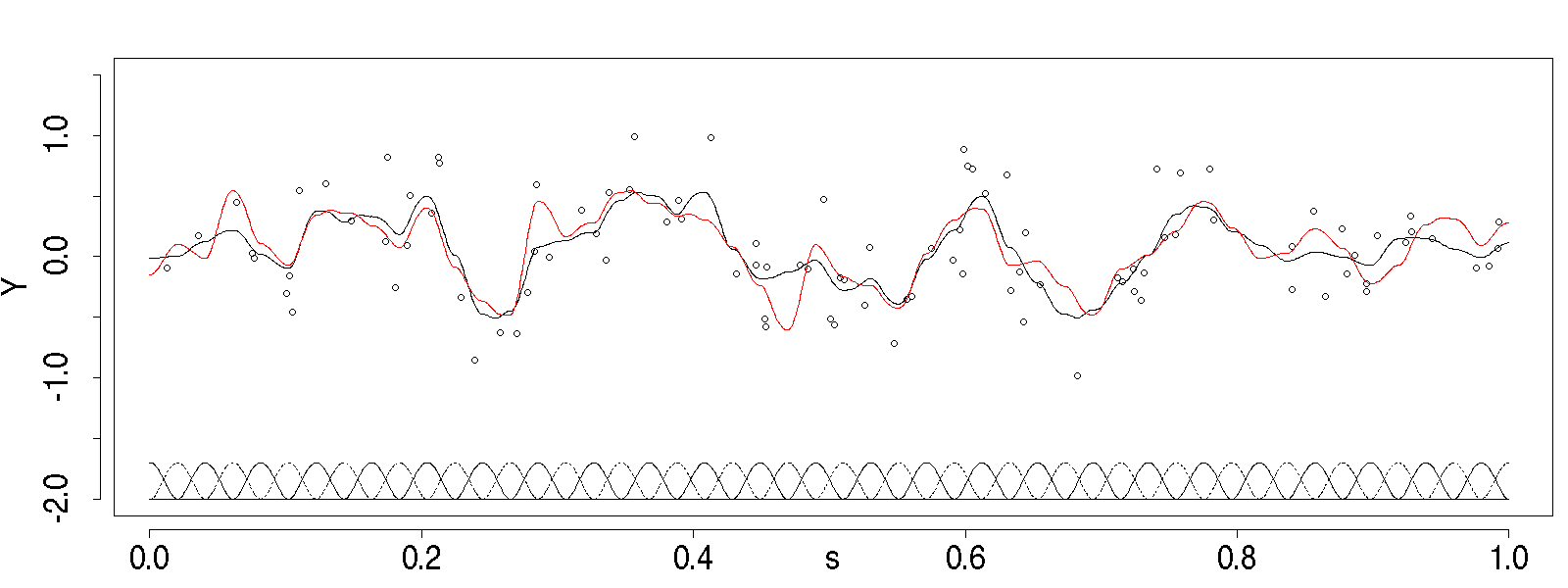}  
		\caption{A realisation of $Y(\cdot)$ (red line) is observed at random locations in the presence of noise (black circles). The data is modelled using a set of bisquare functions (shifted to $Y = -2$ in the figure) with a second-order CAR prior distribution on their coefficients. A Gaussian update gives the conditional mean (black).} \label{fig:exp_setup}
	\end{center}
\end{figure}

\begin{figure}[!t]
	\begin{center}
		\includegraphics[width=\textwidth]{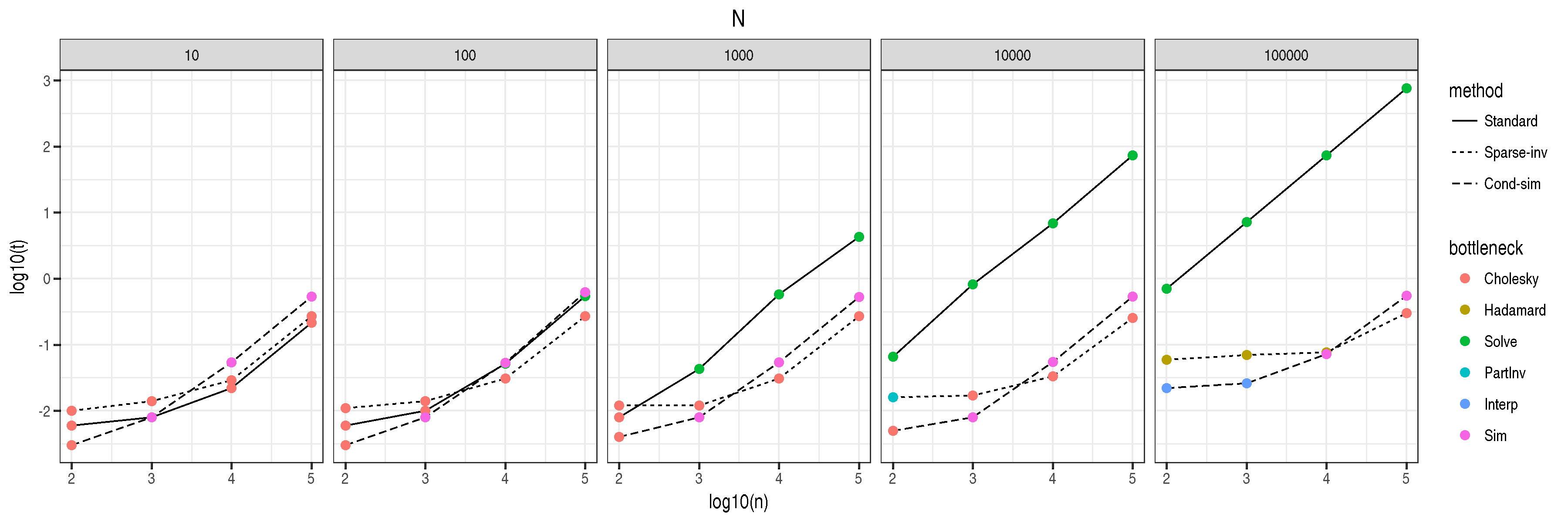}  
		\caption{Time needed in $\log_{10}$ seconds for computing $N$ prediction variances from an $n$-dimensional second-order CAR model observed with $m = 10000$ observations for the direct method (`Standard,' solid line),  the sparse inverse subset method (`Sparse-inv,' short-dashed line) and conditional simulation (`Cond-sim,' long-dashed line). The colour of the bullet indicates which operation was found to be the most computationally-intensive for a given $n$, $N$, and method; see text for details.} \label{fig:Sim1}
	\end{center}
\end{figure}

We compare the time required to compute the prediction variances as a function of $N$ and $n$ using these three methods in Figure~\ref{fig:Sim1}. The bullets denote $\log_{10}$ of the total time in seconds, while the colours denote the operation in the respective algorithm which required the most time to compute: `Cholesky,' the Cholesky decomposition including permutations (although none are needed in this case since $\Qmat$ is banded); `Hadamard,' the Hadamard product $\A \circ (\A\tilde\S)$ in the sparse inverse approach; `Solve,' the forward substitution $\Lmat\Gmat = \Amat\trans$ in the direct method; `PartInv,' computation of the Takahashi equations; `Interp,' the multiplication of $\Amat$ with the simulation ensemble for finding the empirical prediction standard error; and `Sim,' the time required to generate the $\{\wvec_i\}$ and perform the backwards substitution in conditional simulation to find the $\{\vvec_i\}$.

From the figure it is apparent that the direct method of computation performs satisfactorily for small state dimension ($n \le 10^3$) and for a small number of prediction locations ($N \le 10^3$) but as expected becomes intractable as $n$ and $N$ grow.  In the largest example considered ($n = 10^5$ and $N = 10^5$), direct computation required two orders of magnitude more computation time. The most costly operation in the direct method here is seen to be the forward solve $\Lmat\Gmat = \A^{\trans}$ which involves sparse (but still relatively dense) $10^5 \times 10^5$ matrices. This operation can be parallelised, and therefore an equivalent conclusion is that one would need to parallelise a hundred-fold to compute with the direct method what can be obtained serially in the same amount of time using the sparse inverse subset.  In all cases, conditional simulation using 50 simulations is seen to require roughly the same order of magnitude of amount of time as the method involving the sparse inverse subset. For larger state spaces, the bottleneck for the sparse inverse subset method is the sparse Cholesky decomposition, which is an unavoidable operation in all the methods we consider here.

Since $\Qmat$ is banded (with $b=2$), the conclusions from Section \ref{sec:complexity} can be readily observed in the results. First, for the direct method, `Solve' is dominating the computation for $N \ge 1000$. This cost is linear in both $n$ and $N$ as evident from Fig.~\ref{fig:Sim1}. Second, for the sparse inverse subset, when $N = 100000$ the Hadamard product dominates the computation for $n \le 10000$. The complexity of this computation, however, is independent of $n$, and remains roughly constant until the time required to compute the Cholesky decomposition dominates. %This, in turn, scales as $O(n)$  for fixed bandwidth as can be seen in the $N = 100$ case, where the cost of the Cholesky decomposition dominates for all $n \ge 1000$. 
Finally, when conditionally simulating, `Sim' is the bottleneck. The back substitution  here scales as $O(n)$ for fixed bandwidth, as can be seen from Fig.~\ref{fig:Sim1}. In summary, all methods scale as $O(n)$ as expected, but since the Cholesky decomposition dominates the time required in the sparse inverse subset approach for large $n$ and $N$, it outperforms the other two methods in this regime (since the decomposition is also required by the other two).

In this study we let the number of conditional simulations, $M = 50$. This seems reasonable, but was found to yield inaccurate estimates of the prediction standard errors. In the next section we consider the time-accuracy tradeoff of conditional simulation in more detail when comparing it to use of the sparse inverse subset.

%can also be prohibitively inaccurate in certain applications. Let $\hat \dvec$ be the variances estimated through conditional simulation in one experimental run, and $\dvec$ the true variances. Then the interquartile range of the standardised residuals $(\hat\dvec - \dvec)/\dvec$ give an indication of the accuracy of the conditional simulation method for estimating variances. In our case, this range was 27\% for $n_{\SIM} = 50$, 8\% for $n_{\SIM} = 500$ and 2.5\% for $n_{\SIM} = 5000$. This suggests that orders of magnitude more simulations might be needed to decrease the Monte Carlo of the variance error to acceptable levels. Such considerations are of course not needed with the sparse inverse approach. We also observe this in the study using real data and using the popular package \texttt{LatticeKrig} in the next section.

\subsection{Computing prediction standard errors with \texttt{LatticeKrig} models}\label{sec:SST}

\texttt{LatticeKrig} is an \texttt{R} package designed for the modelling and prediction of very large datasets \citep{Nychka_2015}. It constructs a set of multi-resolution basis functions regularly distributed in the domain of interest, and models the conditional dependencies between the basis function coefficients $\etab$ using a sparse precision matrix. Due to sparsity, the total number of basis functions can be large, up to a few hundreds of thousands. Currently, conditional simulation is used to estimate prediction variances at arbitrary prediction locations; this is both time consuming and relatively inaccurate, as we show below.

In this example we consider sea-surface temperature (SST) data taken from the Visible Infrared Imaging Radiometer Suite (VIIRS) on board the Suomi National Polar-orbiting Partnership (Suomi NPP) weather satellite on October 14 2014 \citep{Cao_2013}. The VIIRS sensor reads sea-surface temperatures at very high fidelity, but the data is irregularly distributed spatially. \texttt{LatticeKrig} is ideally suited to supply predictions and prediction standard errors at high resolution from such a spatially-referenced large dataset.

For illustration, we consider a thinned sample of 27745 SST data points in a window $D$ spanning $165^\circ$W--$145^\circ$W and $10^\circ$S--$10^\circ$N. Since we are mostly interested in anomalies, we subtract the empirical mean to yield the data shown in Figure \ref{fig:SST}, left panel. We configured \texttt{LatticeKrig} to construct 26538 basis functions, $\phib(\cdot)$, across three resolutions, and once again assumed that there are no covariates, and that the data have point support. We then fitted the model using the anomalies and produced a high-resolution 1000 $\times$ 1000 gridded map of prediction standard errors in $D$ (i.e., $N = 10^6$).  Hence, once again, $\cL_k^O, k =1,\dots,m,$ and $\cL_j^P,j = 1,\dots,N,$ are identity operators. Fitting the model took 90 seconds, while generating 100 conditional simulations (serially) to generate the map of prediction standard errors required just over 17 minutes. 

The \texttt{LatticeKrig} model places, a priori, an independent GMRF on each of the resolutions, and hence $\Qmat$ is block diagonal. Each prediction location is, however, a linear combination of the basis functions across the different resolutions and unlike in Section \ref{sec:simstudy} Condition \eqref{eq:cor2} is not implicitly satisfied. Condition \eqref{eq:cor1} is also not satisfied as a result of there being large gaps in the data. We therefore compute $\ones(\Amat\trans\Amat)$ and add 1s where needed to $\ones(\Qmat)$ such that Corollary \ref{cor:conditions} (ii) is satisfied. The sparse inverse subset required slightly more time than 100 conditional simulations (nearly 20 minutes) but, more importantly, provided exact prediction standard errors, shown in Figure \ref{fig:SST}, centre panel. 

To see the improvement in accuracy over conditional simulation, we study how the relative error of the prediction-standard-error estimate using conditional simulation decreases with the number of conditional simulations, $M$. Specifically, we define the relative error as
$$
R_i(M) = \frac{\hat{\sigma}^{(M)}_{\textrm{cond},i} - \sigma_{\textrm{true},i}}{\sigma_{\textrm{true},i}},
$$
where $\hat{\sigma}^{(M)}_{\textrm{cond},i}$ is the sample prediction standard error estimated using $M$ conditional simulations, while $\sigma_{\textrm{true},i} \ldef \sqrt{d_i}$ is the exact prediction standard error. The summaries 
\begin{align*}
\hat{r}(M) &= \frac{1}{N}\sum_{i=1}^N R_i(M), \\
\hat{s}(M) &= \sqrt{\frac{1}{N-1}\sum_{i=1}^N \left(R_i(M) - \hat{r}(M)\right)^2},
\end{align*}
then reveal the bias and spread in the relative error, respectively. 

In Figure \ref{fig:SST}, right panel, we show plots for $\hat{r}(M)$ and $\hat{s}(M)$ for $M$ varying between 10 and 100 in intervals of 10 simulations. As expected the empirical mean of the relative error decreases with $M$, and is less than 1\% by $M=30$. However, we also observe that the empirical standard deviation of the relative error does not decrease as rapidly, and it is 10\% at $M=50$ and still around 7\% at $M = 100$. This empirical standard deviation decreases as $1/\sqrt{M}$, and in this scenario one would need to carry out several hundred more conditional simulations in order to bring this quantity down to more acceptable levels of, say, 1\%. The approach involving the sparse inverse subset is hence particularly attractive since it needs the same amount of time required to generate  around 120 conditional simulations.

\begin{figure}[!t]
	\begin{center}
		\includegraphics[width=\textwidth]{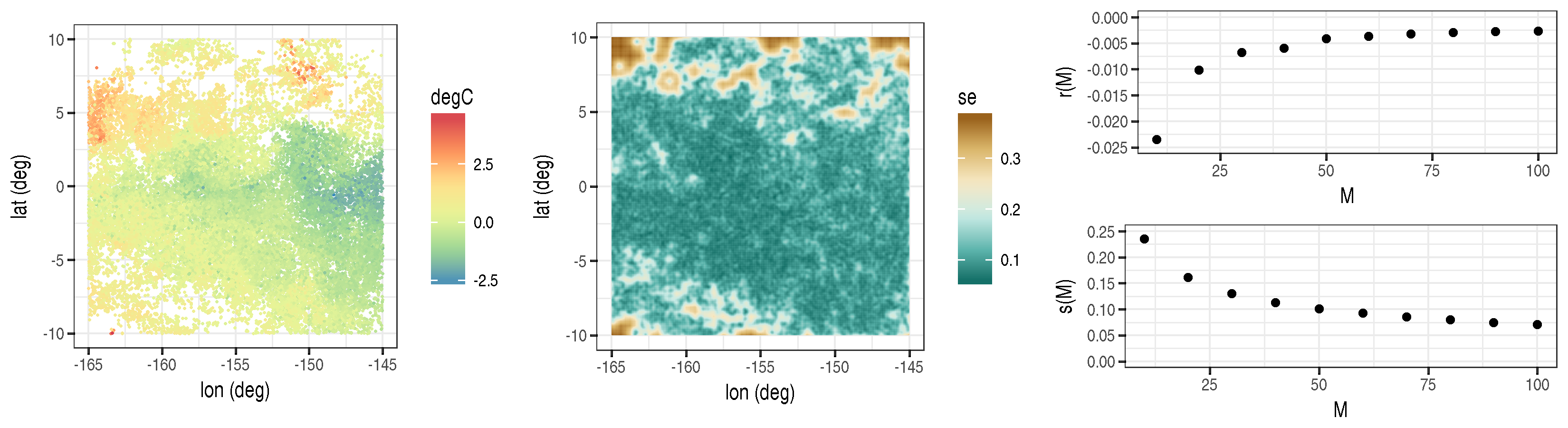}  
		\caption{Computing prediction variances using \texttt{LatticeKrig}. (Left panel) SST data from the VIIRS sensor on October 15 2014. (Centre panel) Prediction standard errors computed using the sparse inverse subset. (Right panel) The empirical mean of the relative error $\hat{r}(M)$ (top) and the empirical standard deviation of the relative error $\hat{s}(M)$ (bottom) when using $M$ conditional simulations from a \texttt{LatticeKrig} model with the SST dataset. } \label{fig:SST}
	\end{center} 
\end{figure}

\subsection{Statistical downscaling} \label{sec:downscaling}

Statistical downscaling is used to make inference on a process of interest on a support that is smaller than that of the measurement. In official statistics, this problem is classified as a branch `small-area estimation,' where one aims to infer statistics of sub-populations from surveys aggregated over large regions. The challenge in statistical downscaling is not so much the small-area prediction as the prediction standard errors that need to be attributed to the small area. 

The Australian Statistical Geography Standard considers a series of nested geographical areas in Australia known as Statistical Area Levels. At the lowest level, Mesh Blocks contain between 30 to 60 dwellings each, Statistical Area Level 1 (SA1) regions have an average population of 400 people each, and the Statistical Area Level 2 (SA2) regions have an average population of between 3000 and 25000 people each. We consider a region of New South Wales containing 78883 Mesh Blocks,  13830 SA1 regions, and 379 SA2 regions, and aim to infer the weekly mean family income ($\MFI$) of earners earning below $\$4000$ per week at the SA1 level just from data at the SA2 level collected in the Census of 2011. We also have $\MFI$ at the SA1 level, but we leave this out to validate the inferred prediction standard errors obtained using only the coarser resolution data. We do not consider SA2 regions where the MFI is based on a response of less than 100 individuals.

As basis functions we choose the Mesh Blocks, hence
$$
\phi_i(\svec) = \begin{cases} 
1 & \svec \in D_i \\
0 & \textrm{otherwise},
\end{cases}
$$
where $D_i \subset D$ is the spatial footprint of the $i$th Mesh Block and $\cup_i D_i = D$. We use a first-order CAR prior to model $\MFI$ at the Block Mesh level, where we can expect high spatial correlation between neighbouring blocks. Specifically, we let $\etab \sim \Gau(\zerob,\Qmat^{-1})$ where $\Qmat \ldef \tau \Dmat_w(\Imat -  \rho\widetilde\Wmat)$. The matrix $\Dmat_w$ is diagonal with $[\Dmat_w]_{jj}$ equal to the number of Mesh Blocks adjacent to the $j$th Mesh Block, while the matrix $\widetilde{\Wmat}$ is given by
$$
\widetilde{W}_{ij} = \begin{cases} 
[\Dmat^{-1}_w]_{jj} & j\sim i \\
0 & \textrm{otherwise},
\end{cases}
$$
\noindent where here $j \sim i$ indicates that the $j$th Mesh Block shares a common border with the $i$th Mesh Block.  More details on this representation can be found in \citet[Section 4.3]{Banerjee_2015}.

As covariates we used an intercept, and the proportion of people aged 15 years or over who are not in school that have completed Year 12 of study (as a proxy for education). Education data is available at both the SA1 level and the SA2 level. Since this is also required at the Mesh Block level when modelling, we assumed that the proportion of people that have reached Year 12 in a Mesh Block is the same as that at the corresponding SA1 level.  We modelled the measurement-error precision as $\Rmat = \sigma^{-2}\Vmat^{-1}$, where $\Vmat \ldef \diag(\{1/{\tilde{n}_k} : k = 1,\dots,m\})$ and where $\tilde{n}_k$ is the total number of respondents in SA2 area $k$.  This heteroscedastic model is based on the assumption that the standard error of the sample mean decreases as $1/\sqrt{\tilde{n}_k}, k = 1,\dots,m$; this is a reasonable assumption in practice \citep[e.g.,][]{Burden_2015}.

In this case study each datum has a footprint $D_k^O \subset D$. Therefore
\begin{equation}
\cL_k^O\cdot Y \ldef \frac{\int_{D^O_k}\psi(\svec)Y(\svec)\intd\svec}{\int_{D^O_k}\psi(\svec)\intd\svec},~k = 1,\dots,m, \label{eq:Ybar}
\end{equation}
where $\psi(\svec)$ is a spatial-weighting function, which is not equal to 1 since Mesh Blocks containing more residents have a bigger influence on the average at the area level than Mesh Blocks with a lower number of residents. The density function $\psi$ is a population density function in units of residents per unit area. To approximate the integrals in \eqref{eq:Ybar} we discretise by Mesh Block, to obtain
$$
\cL_k^O \cdot Y \approx \frac{\sum_{i=1}^n1^{(2)}_{ki}\psi_iY_i|D_i|}{\sum_{i=1}^n1^{(2)}_{ki}\psi_i|D_i|} =  \frac{\sum_{i=1}^nw_i1^{(2)}_{ki}(\xvec_i\trans\betab + \eta_i)}{\sum_{i=1}^nw_i1^{(2)}_{ki}},
$$
\noindent where the indicator function $1^{(2)}_{ki}$ is one if the $i$th Mesh Block is inside the $k$th SA2 region and zero otherwise, $Y_i$ is the MFI, $\psi_i$ is the population density, $|D_i|$ is the area, $w_i = \psi_i|D_i|$ is the resident population, $\xvec_i$ are the covariates (intercept and education), and $\eta_i$ is the basis function coefficient associated with the $i$th Mesh Block, respectively. Hence, in this case, $\Xmat^O \ldef \Bmat\Xmat$ where $\Xmat$ are the covariates at the Mesh Block level, $B_{ki} = w_{i}1^{(2)}_{ki} / \sum_{i=1}^n w_{i}1^{(2)}_{ki}$ and $w_{i}$ is the estimated resident population in the $i$th Mesh Block. The process ($\MFI$) at each Mesh Block is therefore mapped to the $\MFI$ at the SA2 level through the model
$$
\cL_k^O\cdot Y = \beta_1 + \sum_{i=1}^nB_{ki}(\beta_2x_i + \eta_i),
$$
\noindent where $x_i$ is the proportion of people aged 15 years or over in the $i$th Mesh Block who are not in school and that have completed Year 12 of study.

The parameters $\thetab \ldef (\beta_1,\beta_2, \rho,\tau,\sigma^2)\trans$ were estimated using maximum likelihood. To predict at the SA1 level we constructed the matrix $\Amat$ in a similar way to $\Bmat$; specifically we set $A_{ji} = w_{i}1^{(1)}_{ji} / \sum_{i} w_{i}1^{(1)}_{ji}$ where now  $1^{(1)}_{ji}$ is one if the $i$th Mesh Block is inside the $j$th SA1 region and zero otherwise. Prediction over $\Yvec$ proceeds by computing $\hat\Yvec$ and $\dvec$. It is straightforward to see in this case that due to nesting of the Mesh Blocks within SA1 regions, and of the SA1 regions within the SA2 regions, Condition \eqref{eq:cor1} is satisfied. Therefore, the sparse inverse subset can be readily used in the computation of the prediction standard errors.

In Figure \ref{fig:SA2}, left panel, we show the MFI data at the SA2 level; in Figure \ref{fig:SA2}, centre panel, we show the prediction standard error at the SA1 level; and in Figure \ref{fig:SA2}, right panel, we plot the proportion of validation data, $p_2$, that lie in the lower $p_1$ quantiles of the prediction distributions at the SA1 level. This plot demonstrates practically perfect uncertainty quantification, and hence downscaling validity. Since this problem is relatively low-dimensional, the performance between the direct method and the method using the sparse inverse subset was not drastic, their requiring only 20s and 10s respectively for the computation of the prediction variances. As demonstrated in the simulation experiments we can expect larger dividends when modelling over larger geographical regions.

\begin{figure}[!t]
	\begin{center}
		\includegraphics[width=0.36\textwidth]{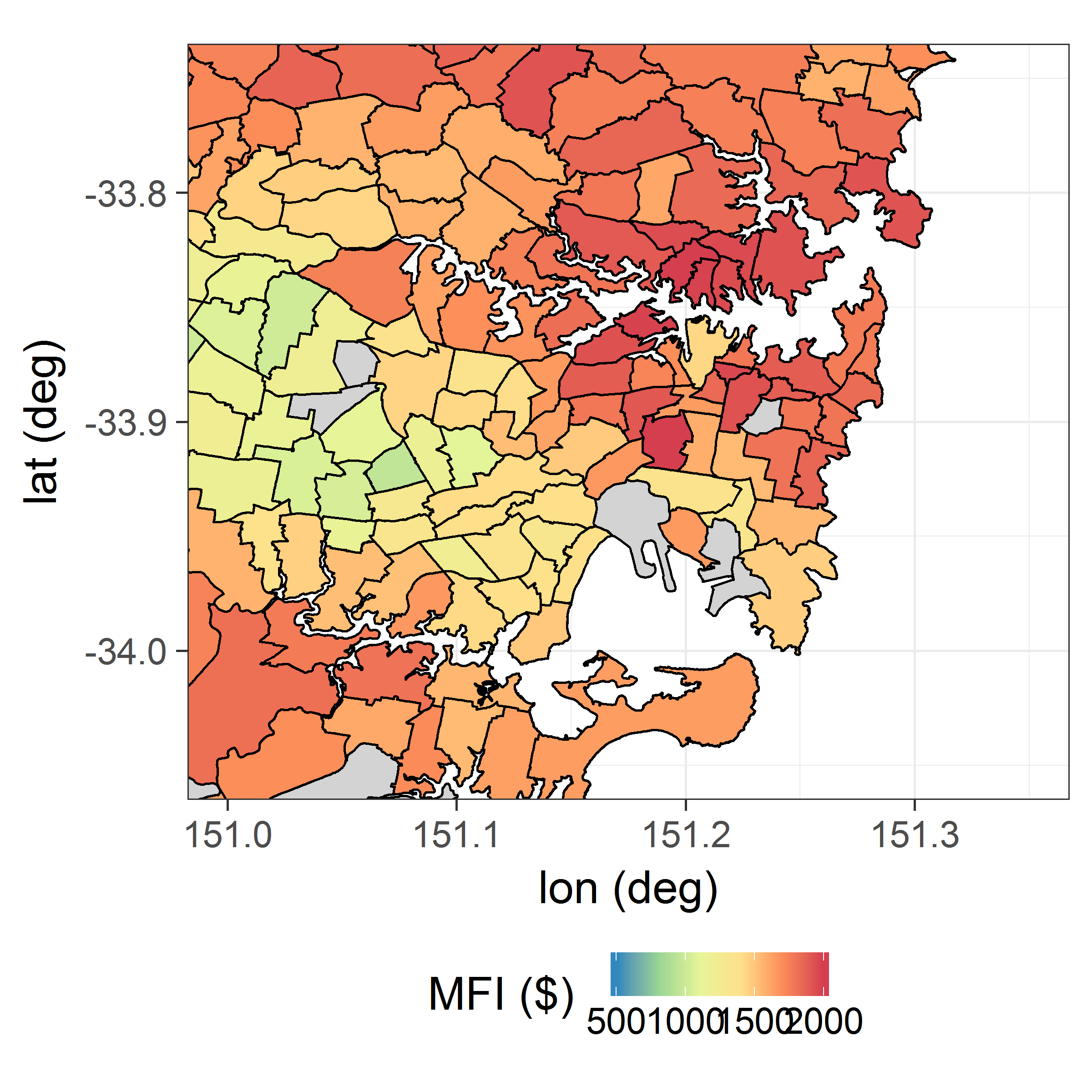}  
		\includegraphics[width=0.36\textwidth]{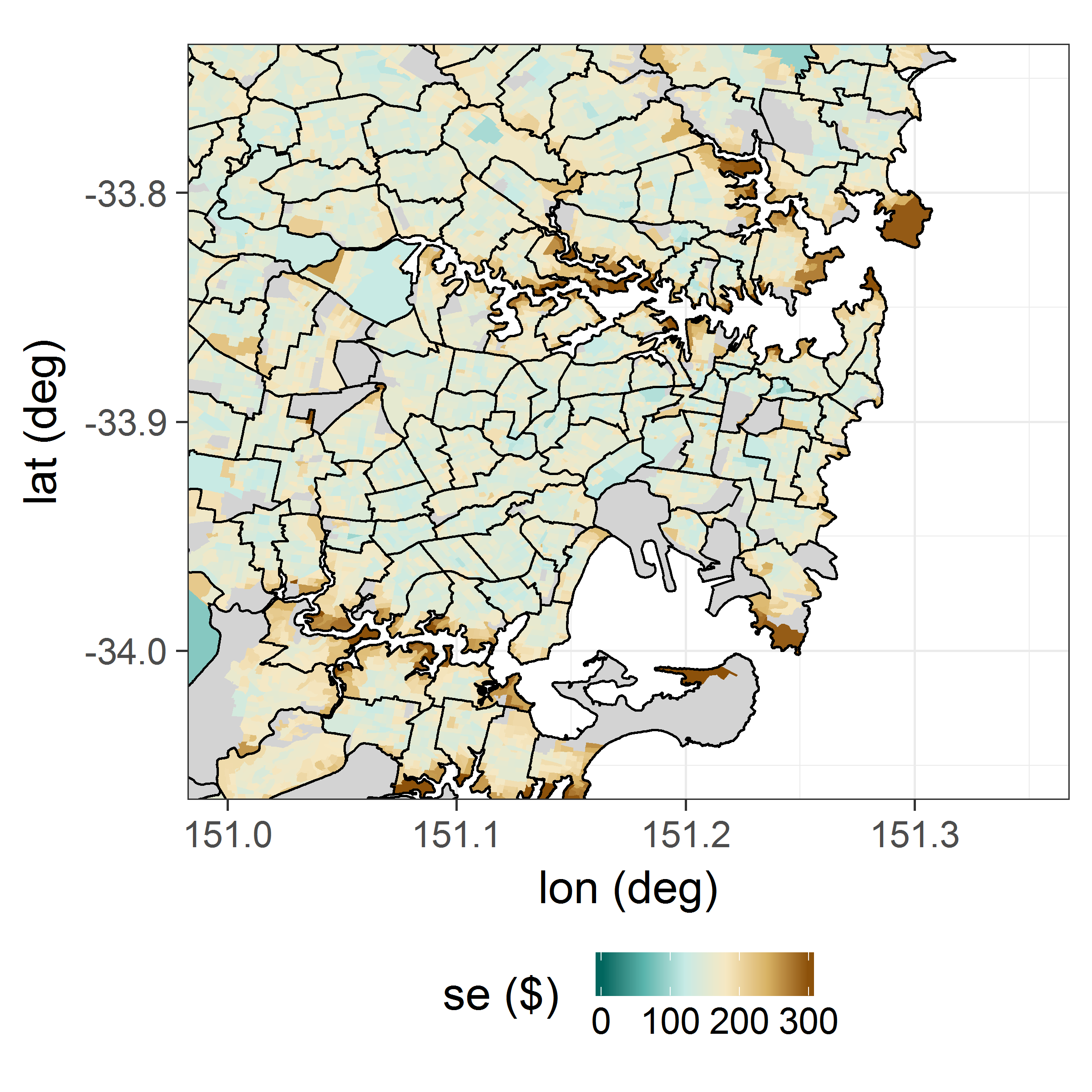}  
		\includegraphics[width=0.26\textwidth]{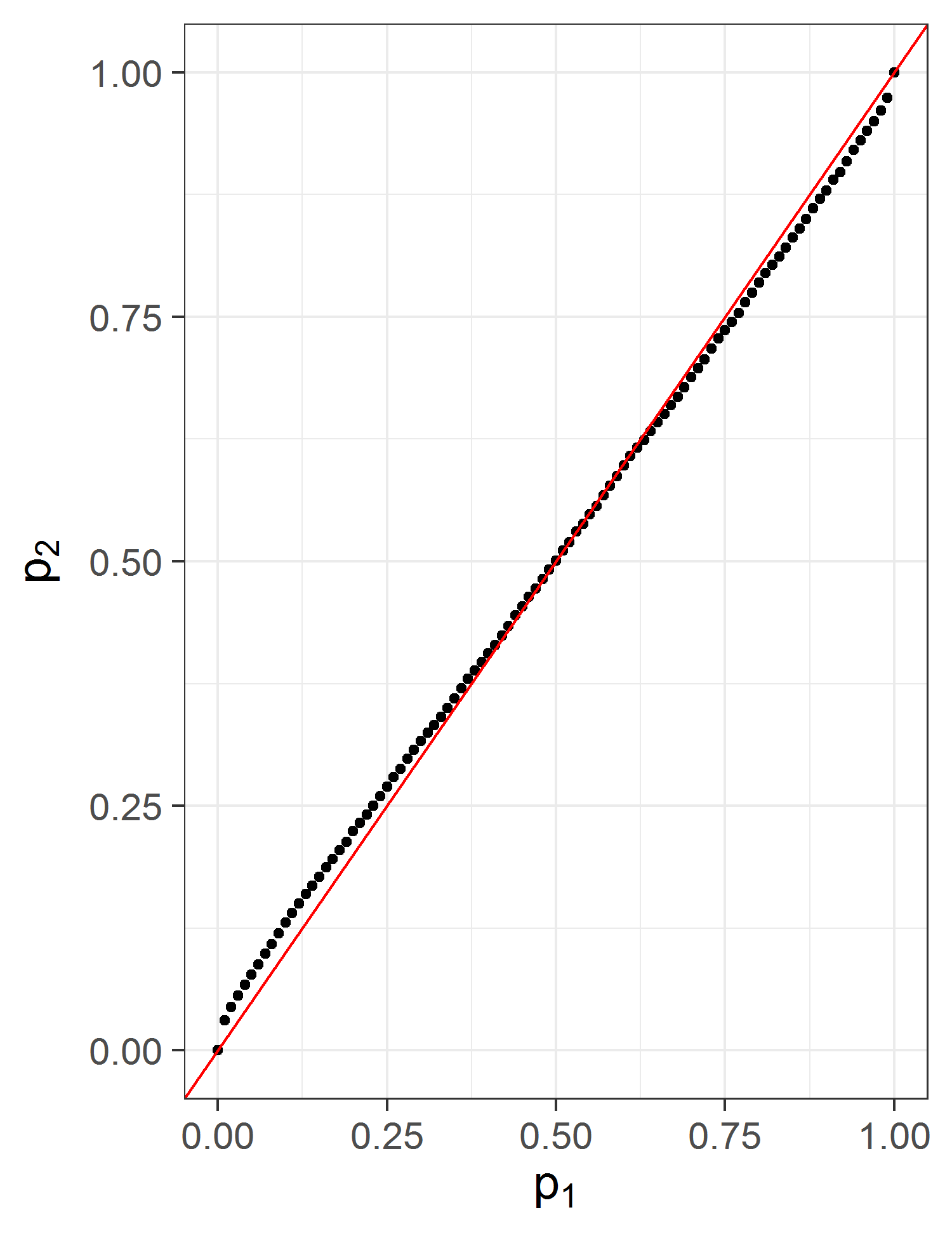}  
		\caption{Statistical downscaling of $\MFI$. (Left panel) Data observed at the Statistical Area 2 (SA2) level. Grey regions denote regions omitted because of paucity in the number of respondents. (Centre panel) Prediction standard errors at the SA1 level. The black borders denote the SA2 levels, while grey areas denote SA1 regions that contain Mesh Blocks without any residents (e.g., airports, parks, etc.). (Right panel) Proportion of validation data points $p_2$ in the lower $p_1$ quantile of the prediction distributions at the SA1 level. The red line denotes the ideal case when $p_1 = p_2$.} \label{fig:SA2}
	\end{center}
\end{figure}

\subsection{Fixed rank kriging with a CAR model for the fine-scale variation}

Fixed rank kriging (FRK) is an optimal spatial prediction methodology that treats the process as a sum of basis functions, the weights of which are modelled using an $n \times n$ covariance matrix $\Kmat$. By virtue of the Sherman--Woodbury--Morrison formula, all estimation and prediction equations in FRK only involve inverses of matrices of size $n\times n$, where $n$ is generally much less than $m$. This allows optimal prediction using large datasets; see \citet{Cressie_2008} for details, and \citet{FRK} for a software implementation.

A drawback of FRK is that the number of basis functions $n$ has to be relatively small, and limited to a few thousand, due to the (dense) matrix inversion required of $\Kmat$ at various stages in estimation and prediction. To cater for this small number, extra variability is captured in FRK through use of a fine-scale variation term $\xi(\cdot)$ which is generally taken to spatially uncorrelated. A more realistic model is one that assumes that $\xi(\cdot)$ has got fine-scale correlations. This could, for example, be included as a CAR model on a fine discretisation of the domain of interest.

Assume, for simplicity, that the CAR model is defined over our prediction grid (this need not be the case, in general). Then the FRK--CAR model for the process evaluated over the elements of the grid is given by
\begin{equation}\label{eq:FRK-CAR}
\Yvec = \Xmat^P\betab + \tilde\Amat\tilde\etab + \xib,
\end{equation}
where $\tilde\etab \sim \Gau(\zerob,\Kmat)$ and $\xib \sim \Gau(\zerob,\Qmat_\xi^{-1})$, where $\Qmat_\xi$ is the precision matrix of the CAR model. Prediction of $\Yvec$ when $\xib$ is spatially uncorrelated is straightforward \citep[e.g., ][]{Katzfuss_2011}, since in that case $\xib$ is independent of the data $\Zvec$ at all unobserved locations. When $\xib$ is spatially correlated this is no longer the case, and $\tilde\etab$ and $\xib$ need to be considered jointly. Let $\etab \ldef (\tilde\etab\trans,\xib\trans)\trans$ and $\Amat \ldef (\tilde \Amat,\Imat)$. Then \eqref{eq:FRK-CAR} can be re-written as $\Yvec = \Xmat^P\betab + \Amat\etab$. As usual, the data $\Zvec = \Xvec^O\betab + \Bmat\etab + \epsilonb$ where, similar to $\Amat$, $\Bmat = (\tilde \Bmat, \Imat)$ is partitioned into two matrices corresponding to $\tilde \etab$ and $\xib$, respectively. In practice we do not know the parameters $\betab$ and those appearing in the matrices $\Rmat, \Kmat$ and $\Qmat_\xi$. In big-data situations it is reasonable to estimate these using maximum likelihood prior to carrying out prediction \citep[e.g.,][]{FRK,Ma_2017}.  

Once these matrices are constructed, everything proceeds as in the earlier examples. For prediction we have that $\dvec = \diag(\Amat \Smat\Amat\trans)$ where $\Smat \ldef (\Bmat\trans \Rmat \Bmat + \Qmat)^{-1}$ where the block-diagonal matrix $\Qmat \ldef \textrm{bdiag}(\Kmat^{-1},\Qmat_\xi)$, and where $\textrm{bdiag}(\cdot)$ returns a block-diagonal matrix of its matrix arguments. We let $\Qmat_\xi = \tau\Dmat_w(\Imat - \rho\widetilde\Wmat)$ as in Section \ref{sec:downscaling}, although any positive-definite sparse precision matrix could be used. 

Before using the sparse inverse subset to compute $\dvec$ we first need to make sure that the conditions of Corollary \ref{cor:conditions} are satisfied. In this example,
$$\Amat\trans\Amat = \begin{bmatrix} \tilde{\Amat}\trans\tilde{\Amat} & \tilde{\Amat}\trans \\ \tilde\Amat & \Imat \end{bmatrix}; \quad \Bmat\trans\Bmat = \begin{bmatrix} \tilde{\Bmat}\trans\tilde{\Bmat} & \tilde{\Bmat}\trans \\ \tilde\Bmat & \Imat \end{bmatrix}; \quad \Qmat = \begin{bmatrix} \Kmat^{-1} & \zerob \\ \zerob & \Qmat_\xi \end{bmatrix}.$$
Now, data typically used in FRK is large and typically irregular in space, so that, for the same reasons as in Section \ref{sec:SST}, it is unlikely that $\ones(\Bmat\trans\Bmat) \ge \ones(\Amat\trans\Amat)$. The matrix $\Kmat$ on the other hand, is low dimensional and typically assumed to be dense, therefore trivially $\ones(\Kmat^{-1}) \ge \ones(\tilde\Amat\trans\tilde\Amat)$. Also trivially, $\ones(\Qmat_\xi) \ge \ones(\Imat)$. Therefore, in order to ensure that Condition \eqref{eq:cor2} holds we simply need to force zeroes in the lower off-diagonal block and the upper off-diagonal block of $\ones(\Qmat)$  to have 1s where $\tilde{\Amat}$ and $\tilde{\Amat}\trans$ are nonzero, respectively. Once these 1s are inserted, the symbolic Cholesky factor will have the required sparsity structure to satisfy Corollary \ref{cor:conditions} (ii) and the sparse inverse subset can be used to compute $\dvec$.

\begin{figure}[!t]
	\begin{center}
		\includegraphics[width=0.78\textwidth]{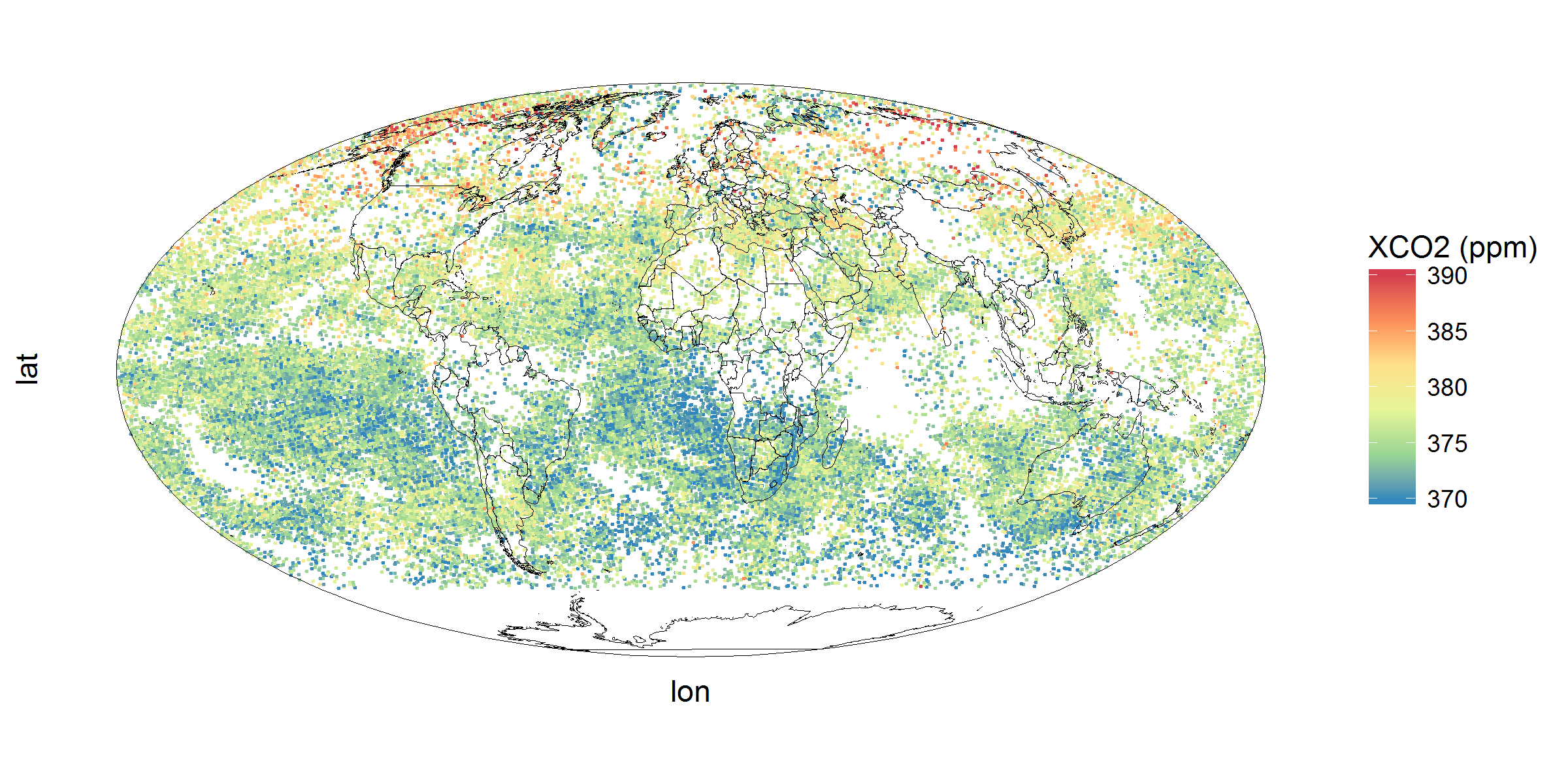}  
		\includegraphics[width=0.78\textwidth]{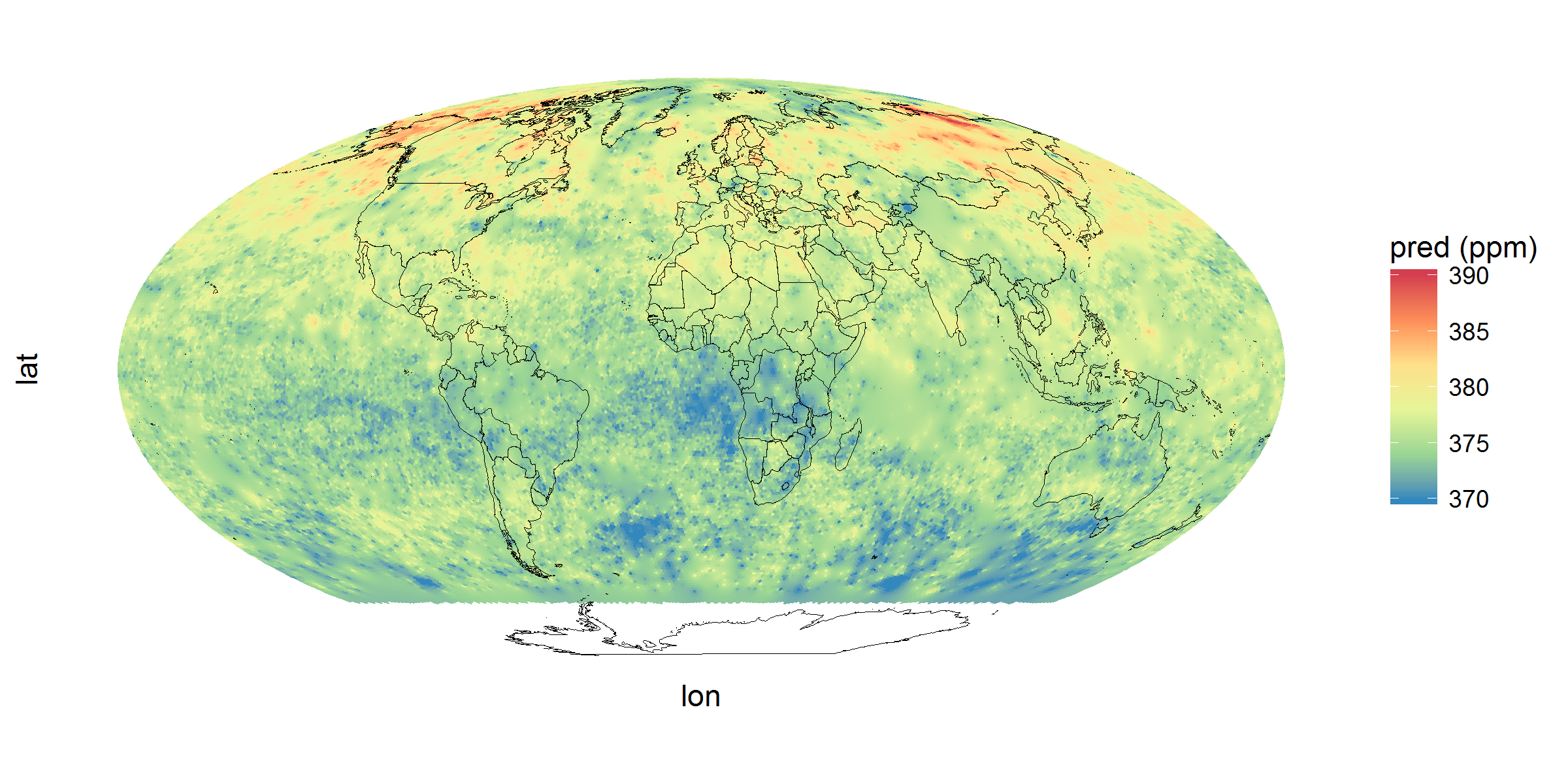}  
		\includegraphics[width=0.78\textwidth]{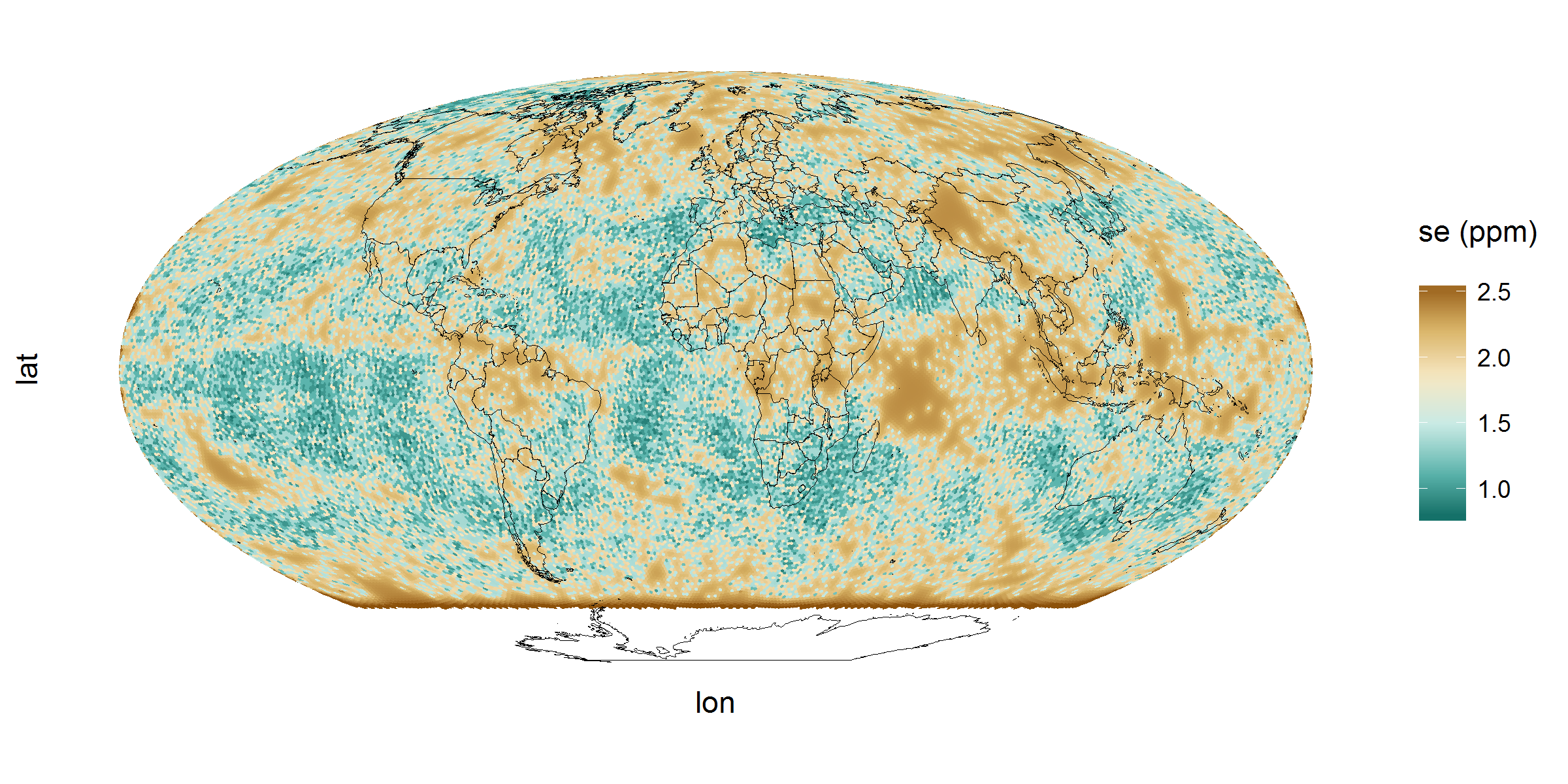}  
		\caption{(Top panel) XCO$_2$ data in ppm from the AIRS instrument between May 01 2003 and May 03 2003 (inclusive). (Centre panel) Prediction of $\Yvec$ in ppm using an FRK-CAR model. (Bottom panel) Prediction standard error of $\Yvec$ in ppm using an FRK-CAR model. Note that AIRS does not release data below 60$^\circ$S.} \label{fig:FRK}
	\end{center}
\end{figure}

We demonstrate the use of the FRK-CAR model on 43059 column-averaged carbon dioxide retrievals (XCO$_2$) generated by the Atmospheric Infrared Sounder (AIRS) on board the Aqua satellite between May 01 2004 and May 03 2004 inclusive. We constructed the prediction grid using an aperture 3 hexagonal discrete global grid, consisting of 62424 equal-area polygons on the sphere. As basis functions we used three resolutions of 1176 bisquare basis functions arranged regularly on the sphere. We used an intercept and the latitude coordinate as covariates to construct $\Xmat^P$ and subsequently fitted the model by estimating $\betab$, $\tau$, $\rho$, and the parameters appearing in $\Kmat$ using standard maximum-likelihood. The standard deviation of the measurement error was fixed to 2 ppm.

In Figure \ref{fig:FRK}, top panel, we show the data, in the centre panel we show the predictions, and in the bottom panel we show the resulting prediction standard errors computed using the sparse inverse subset after padding $\ones(\Qmat)$ with 1s as described above. Note the fine-scale structure now apparent in the prediction standard errors as a result of incorporating a CAR model in FRK; such structure is usually not present with standard FRK models. In this example, following the Cholesky decomposition of $\Pmat$, the direct method for computing $\dvec$ required over two minutes whereas that using the sparse inverse subset required just under 20 seconds.

\section{Conclusion}\label{sec:conc}

In this article we explore the use of the sparse inverse subset obtained using the Takahashi equations for computing prediction standard errors of linear combinations of $\etab$, where $\etab$ is a Gaussian Markov random field. We derive sufficient conditions that depend on the sparsity structure of the matrices involved, and carry out an extensive simulation study illustrating that this approach is feasible in high dimensional problems when a large number of linear combinations are sought. We specifically focus on case studies in spatial statistics, where the spatial process is decomposed using a sum of basis functions, to demonstrate its utility; it is here where we envision this approach to find most use. In general, we find that the method using the sparse inverse subset considerably outperforms conditional simulation, and provides the required prediction standard errors where direct methods would require a high degree of parallelisation.

In this work we have largely neglected variability introduced through the fixed effects. This was intentional: The problems we considered have sufficiently large data sets that any variability introduced through estimating the fixed effects is negligible, and the sparse subset should not be used when this is deemed to pose a problem. One may recast the fixed effects into random effects, thereby treating them as basis functions in the spatial process. However, doing so would destroy the sparsity in $\Pmat$, nullifying any potential advantage of utilising the sparse inverse subset. 

All methods considered in this article require the computation of the Cholesky factor. Of course, in high-dimensional systems where $n$ ventures into the tens of millions, this will not be possible. Approximate methods are required in this scenario; these could be based on model approximations \cite[e.g.,][]{Katzfuss_2017}, likelihood approximations \citep{Eidsvik_2014}, or sample approximations \citep{Simpson_2013}. Recently \citet{Siden_2017}, showed how to approximate covariances by using the Takahashi equations on connected subsets of the GMRF (and hence by computing the Cholesky factor of several, smaller GMRFs) using a Rao-Blackwellised sampling scheme. However, the majority of applications do not require the consideration of such a high latent dimension. Further, the possibility for the practitioner to not have specialist equipment, such as graphical processing units or a high-performance environment, is attractive. The latter, if available, are likely to render conditional simulation and the use of the direct method once again feasible.

\section*{Acknowledgements}

We gratefully acknowledge Botond Cseke for discussions on sparse inverse subsets, Noel Cressie for discussions on Census data, and Yulija Marchetti for supplying the SST data. Data at the Mesh Block, SA1, and SA2 levels used in the study of Section \ref{sec:downscaling} were obtained from the Australian Bureau of Statistics web pages \citep{ABS_MBs,ABS_shapefiles,ABS_DataPacks}.  This research did not receive any specific grant from funding agencies in the public, commercial, or not-for-profit sectors.

\section*{References}

\bibliography{Spat_bib}
\end{document}